\documentclass[11pt,letterpaper,twoside,reqno,nosumlimits]{amsart}

\allowdisplaybreaks

\usepackage[usenames,dvipsnames]{xcolor}
\usepackage{fancyhdr}
\usepackage{amsmath,amsfonts,amsbsy,amsgen,amscd,mathrsfs,amssymb,amsthm}
\usepackage{subfig}
\usepackage{url}

\usepackage{mathtools}

\usepackage{bm,bbm}
\providecommand{\mathbold}[1]{\bm{#1}}
\renewcommand{\mathbb}[1]{\mathbbm{#1}}
\newcommand{\mtx}[1]{\mathbold{#1}}
\newcommand{\Id}{\mathbf{I}}

\newcommand{\RR}{\mathbbm{R}}
\newcommand{\E}{\mathbbm{E}}
\newcommand{\zeromtx}{\bm{0}}
\newcommand{\bone}{\mathbbm{1}}
\newcommand{\ep}{\varepsilon}
\providecommand{\eps}{\ep}
\newcommand{\econst}{\mathrm e}
\usepackage[notcite,notref,final]{showkeys}
\newcommand{\Prob}[1]{\mathbbm{P}\left\{{#1}\right\}}
\DeclareMathOperator{\rank}{rank}

\def\todoc/{\textcolor{red}{\textbf{[citation]}}}

\newcommand{\lawa}{\mathrm P_{\mtx A}}
\newcommand{\stie}[2]{\mathcal S_{#1, #2}}
\newcommand*{\barrho}{\bar \rho}
\newcommand*{\qr}[1]{\mathsf{QR}[#1]}
\newcommand*{\subspacedist}[2]{\mathrm{dist}(#1, #2)}
\newcommand*{\gap}{\rho_k}
\newcommand*{\norm}[1]{\left\|#1\right\|}
\newcommand*{\good}{\mathcal G}
\newcommand*{\filtration}{\mathcal F}

\mathtoolsset{showonlyrefs}

\usepackage[font=small,margin=0.25in,labelfont={sc},labelsep={colon}]{caption}

\usepackage{tikz}
\usepackage{microtype}
\usepackage{enumitem}

\definecolor{dark-gray}{gray}{0.3}
\definecolor{dkgray}{rgb}{.4,.4,.4}
\definecolor{dkblue}{rgb}{0,0,.5}
\definecolor{medblue}{rgb}{0,0,.75}
\definecolor{rust}{rgb}{0.5,0.1,0.1}

\RequirePackage{doi}
\usepackage{hyperref}

\hypersetup{urlcolor=rust}
\hypersetup{citecolor=dkblue}
\hypersetup{linkcolor=dkblue}

\usepackage{setspace}

\usepackage{graphicx}
\usepackage{booktabs,longtable,tabu} 
\usepackage{multirow} 

\usepackage{float}

\usepackage[full]{textcomp}

\usepackage[scaled=.98,sups,osf]{XCharter}
\usepackage[scaled=1.04,varqu,varl]{inconsolata}
\usepackage[type1]{cabin}
\usepackage[charter,vvarbb,scaled=1.07]{newtxmath}
\usepackage[cal=boondoxo]{mathalfa}
\usepackage[T1]{fontenc}

\usepackage{bm}

\graphicspath{{figures/}}

\newtheorem{theorem}{Theorem}[section]
\newtheorem{lemma}[theorem]{Lemma}

\newtheorem{proposition}[theorem]{Proposition}

\theoremstyle{definition}

\newtheorem{assumption}[theorem]{Assumption}

\numberwithin{equation}{section} 
\numberwithin{figure}{section}
\numberwithin{table}{section}

\floatstyle{plaintop}
\newfloat{recipe}{thp}{lor}
\floatname{recipe}{Recipe}
\numberwithin{recipe}{section}

\providecommand{\mathbold}[1]{\bm{#1}}

\renewcommand{\phi}{\varphi}

\renewcommand{\mid}{\mathrel{\mathop{:}}}

\providecommand{\mathbbm}{\mathbb}

\newcommand{\triplenorm}[1]{{\left\vert\kern-0.25ex\left\vert\kern-0.25ex\left\vert #1
    \right\vert\kern-0.25ex\right\vert\kern-0.25ex\right\vert}}

\begin{document}

\title{Streaming $k$-PCA: Efficient guarantees for Oja's algorithm, beyond rank-one updates}
\author[Huang et al.]{De Huang, Jonathan Niles-Weed, and Rachel Ward}
\thanks{The authors gratefully acknowledge the funding for this work.
DH was in part supported by NSF Grants DMS-1907977, DMS-1912654, and the Choi Family Postdoc Gift Fund.
JNW was supported under NSF grant DMS-2015291.
JNW and RW were supported in part by the Institute for Advanced Study,
where some of this research was conducted.
RW received support from AFOSR MURI Award N00014-17-S-F006 and NSF grant DMS-1952735 .}
\date{6 February 2021}

\begin{abstract}
We analyze Oja's algorithm for streaming $k$-PCA, and prove that it achieves performance nearly matching that of an optimal offline algorithm.
Given access to a sequence of i.i.d.~$d \times d$ symmetric matrices, we show that Oja's algorithm can obtain an accurate approximation to the subspace of the top $k$ eigenvectors of their expectation using a number of samples that scales polylogarithmically with $d$.
Previously, such a result was only known in the case where the updates have rank one.  

Our analysis is based on recently developed matrix concentration tools, which allow us to prove strong bounds on the tails of the random matrices which arise in the course of the algorithm's execution.
\end{abstract}


\maketitle

\section{Introduction}
Principal component analysis is one of the foundational algorithms of statistics and machine learning.
From a practical perspective, perhaps no optimization problem is more widely used in data analysis~\cite{Jolliffe2002}.
From a theoretical perspective, it is one of the simplest examples of a non-convex optimization problem that can nevertheless be solved in polynomial time; as such, it has been an important proving ground for understanding the fundamental limits of efficient optimization~\cite{Simchowitz2018}.

In the basic setting, the practitioner has access to a sequence of independent symmetric random matrices $\mtx A_1, \mtx A_2, \dots$ with expectation $\mtx M \in \RR^{d \times d}$.
The goal is to approximate the leading eigenspace of $\mtx M$ or, more generally, to approximate the subspace spanned by its leading $k$ eigenvectors.
While it is natural to attempt to solve this problem by performing an eigen-decomposition of the empirical average $\bar{\mtx A} = \frac 1T \sum_{i=1}^T \mtx A_i$, the amount of space required by this approach can be prohibitive when $d$ is large.
In particular, if the matrices $\mtx A_i$ are sparse or low-rank, performing incremental updates with the matrices $\mtx A_i$ may be significantly cheaper than storing all the iterates or their average.
A tremendous amount of attention has therefore been paid to designing algorithms which can cheaply and provably estimate the subspace spanned by the top $k$ eigenvectors of $\mtx M$ using limited memory and a single pass over the data, a problem known as \emph{streaming PCA}~\cite{JaiJinKak16}.

The simplest and most natural approach to this problem was proposed nearly 40 years ago by Oja~\cite{Oja82:Simplified-Neuron,Oja1985}:
\begin{enumerate}
\item Randomly choose an initial guess $\mtx Z_0 \in \RR^{d \times k}$, and set $\mtx Q_0 \gets \qr{\mtx Z_0}$
\item For $t \geq 1$, set ${\mtx Q}_t \gets \qr{(\Id + \eta_t \mtx A_t) \mtx Q_{t-1}}$.
\end{enumerate}
Here, $\qr{{\mtx Q}_t}$ returns an orthogonal $\RR^{d \times k}$ matrix obtained by performing the Gram--Schmidt process to the columns of ${\mtx Q}_t$.
It is easy to see~\cite[Lemma 2.2]{AllLi17} that the Gram--Schmidt step commutes with the multiplicative update, so that we can equivalently consider a version of the algorithm which performs a single orthonormalization at the end, and outputs
\begin{equation*}
\mtx Q_t = \qr{\mtx Z_t}\,, \quad \mtx Z_t = \mtx Y_t \dots \mtx Y_1 \mtx Z_0\,,
\end{equation*}
where $\mtx Y_i := (\Id + \eta_i \mtx A_i)$.

Oja's algorithm can be viewed as a noisy version of the classic orthogonal iteration algorithm for computing invariant subspaces of a symmetric matrix~\cite[Section 7.3.2]{GolVan96}; alternatively, it corresponds to projected stochastic gradient descent on the Stiefel manifold of matrices with orthonormal columns~\cite{EdeAriSmi99}.
Despite its simplicity and practical effectiveness, Oja's algorithm has proven challenging to analyze because of its inherent non-convexity.

As a benchmark against which to compare Oja's algorithm, we may consider the performance of the simple offline algorithm which computes the leading $k$ eigenvectors of $\bar{\mtx A}$.
We write $\mtx V \in \RR^{d \times k}$ for the orthogonal matrix whose columns are the leading $k$ eigenvectors of $\mtx M$ and $\hat{\mtx V} \in \RR^{d \times k}$ for the matrix containing the leading $k$ eigenvectors of $\bar{\mtx A}$, and measure the quality of $\hat{\mtx V}$ by the following standard measure of distance between subspaces:
\begin{equation*}
\subspacedist{\hat{\mtx V}}{\mtx V} := \|\mtx V \mtx V^* - \hat{\mtx V}\hat{\mtx V}^*\|\end{equation*}

If $\|\mtx A_i - \mtx M\| \leq M$ almost surely and the gap between the $k$th and $(k+1)$th eigenvalues is $\rho_k$, then the Matrix Bernstein inequality~\cite[Theorem 1.4]{Tro12} combined with
Wedin's Theorem~\cite{Wedin1972} implies that there exists a positive constant $C$ such that
\begin{equation}\label{eq:bernstein}
 \subspacedist{\hat{\mtx V}}{\mtx V} \leq C \frac{M}{\rho_k}\sqrt{\frac{\log(d/\delta)}{T}}\,.
\end{equation}
with probability at least $1-\delta$.

The key question is whether Oja's algorithm is able to achieve similar performance.
However, except in the special \emph{rank-one} case where either $k = 1$ or $\rank(\mtx A_i) = 1$ almost surely, no such bound is known.
\subsection{Our contribution}
We give the first results for Oja's algorithm nearly matching~\eqref{eq:bernstein}, for any $k \geq 1$ and updates of any rank.
Our main result (Theorem~\ref{thm:main}) establishes that, after a burn-in period of $T_0 = \tilde O \left(\frac{kM^2}{\delta^2 \rho_k^2}\right)$ steps, the output of Oja's algorithm satisfies
\begin{equation*}
\subspacedist{\mtx Q_T}{\mtx V} \leq C' \frac{M}{\rho_k} \sqrt{\frac{\log(kM/\delta\rho_k)}{T - T_0}}
\end{equation*}
with probability at least $1-\delta$ for a universal positive constant $C'$.
Ours is the first work to show that Oja's algorithm can achieve a guarantee similar to~\eqref{eq:bernstein} beyond the rank-one case.

The assumption that $k = 1$ or $\rank(\mtx A_i) = 1$ is fundamental to the proof strategies used in prior works.
To show that the error decays sufficiently quickly, prior work focuses on the quantity $\|\mtx U^* \mtx Z_t (\mtx V^* \mtx Z_t)^{-1}\|_2$, where the columns of $\mtx U$ are the last $d-k$ eigenvectors of $\mtx M$, which is an upper bound on $\subspacedist{\mtx Q_t}{\mtx V}$.
(See Lemma~\ref{lem:dist_to_w}, below.)
The key challenge is to control the inverse $(\mtx V^* \mtx Z_t)^{-1}$.
When $k = 1$, as in \cite{JaiJinKak16}, this quantity is a scalar, so it can be pulled out of the norm and bounded separately.
This is no longer possible when $k > 1$,  but if $\rank(\mtx A_i) = 1$, as in \cite{AllLi17}, then $\mtx V^* \mtx Z_t$ can be written as a rank-one perturbation of $\mtx V^* \mtx Z_{t-1}$.
The Sherman--Morrison formula then implies that $\mtx U^* \mtx Z_t (\mtx V^* \mtx Z_t)^{-1}$ can be written as $\mtx U^* \mtx Z_{t-1} (\mtx V^* \mtx Z_{t-1})^{-1}$ plus the sum of explicit, rank-one correction terms.
However, if neither $k = 1$ nor $\rank(\mtx A_i) = 1$, this approach quickly becomes infeasible, since the correction terms now involve a product of rank-$k$ matrices whose norm is difficult to bound.

A more subtle difficulty implicit in prior work is that proofs must be carried out entirely in expected (squared) Frobenius norm.
This requirement is necessitated by the fact that the Frobenius norm is Hilbertian, so it is possible to employ the crucial Pythagorean identity
\begin{equation}\label{eq:pythag}
\E\|\mtx Y\|_2^2 = \|\E \mtx Y\|_2^2 + \|\mtx Y - \E \mtx Y\|_2^2
\end{equation}
for any random matrix $\mtx Y$.
It is this identity that makes it possible to control the evolution of $\E \|\mtx U^* \mtx Z_t (\mtx V^* \mtx Z_t)^{-1}\|_2^2$.
However, as our proofs reveal, it is of significant utility to be able to recursively control the operator norm $\|\mtx U^* \mtx Z_t (\mtx V^* \mtx Z_t)^{-1}\|$ with high probability instead.
Unfortunately,~\eqref{eq:pythag} is of no help in proving statements of this kind.

Our argument handles both challenges and represents a significant conceptual simplification over earlier proofs.
Our crucial insight is that, rather than using the squared Frobenius norm, it is possible to prove a stronger recursion in a different norm, which implies high-probability bounds.
Using techniques recently developed by \cite{HuaNilTro20} to prove concentration inequalities for products of random matrices, we show that conditioned on $\|\mtx U^* \mtx Z_{t-1} (\mtx V^* \mtx Z_{t-1})^{-1}\|$ being 
well behaved, the probability that $\|\mtx U^* \mtx Z_{t} (\mtx V^* \mtx Z_{t})^{-1}\|$ deviates significantly from its expectation is exponentially small.

In other words, good concentration properties for $\|\mtx U^* \mtx Z_{t-1} (\mtx V^* \mtx Z_{t-1})^{-1}\|$ imply good concentration properties for the next iterate, $\|\mtx U^* \mtx Z_{t} (\mtx V^* \mtx Z_{t})^{-1}\|$.
These high-probability bounds significantly simplify the calculations, since they allow us to guarantee that the problematic error terms appearing in prior work are small.

If we knew that $\|\mtx U^* \mtx Z_0 (\mtx V^* \mtx Z_0)^{-1}\| = O(1)$ with high probability, then the above induction argument would allow us to conclude that $\|\mtx U^* \mtx Z_{t} (\mtx V^* \mtx Z_{t})^{-1}\| = O(1)$ for all $t$.
Unfortunately, this is not the case: if $\mtx Z_0$ is randomly initialized with i.i.d.~Gaussian entries, then typically
\begin{equation*}
\|\mtx U^* \mtx Z_0 (\mtx V^* \mtx Z_0)^{-1}\| \asymp \sqrt{d k}\,.
\end{equation*}
We therefore adopt a two-phase approach: in the first, short phase, of length approximately $\log d$, we show that the operator norm decays from $O(\sqrt{d k})$ to $O(1)$, and in the second phase we use the above recursive argument to establish that the operator norm decays to zero at a $O(1/\sqrt{T})$ rate.
To simplify the analysis of the first phase, we develop a coupling argument that allows us reduce without loss of generality to the case where the law $\lawa$ of the random matrices $\mtx A_1, \mtx A_2, \dots$ has finite support and obtain almost-sure guarantees by a simple union bound.
This weak control is enough to guarantee that $\|\mtx U^* \mtx Z_t (\mtx V^* \mtx Z_t)^{-1}\|$ decays exponentially fast, so that it is of constant order after approximately $\log d$ iterations.

\subsection{Prior work}
Obtaining non-asymptotic rates of convergence for Oja's algorithm and its variants has been an area of active recent interest~\cite{Sha15,Sha16,Sa2015,LiWanLiu18,LiLinLu16,BalDuWan16,BalDasFre13,HarPri14,JaiJinKak16,MitCarJai13}.
Apart from the results of~\cite{AllLi17} and~\cite{JaiJinKak16}, none of these works proves bounds matching~\eqref{eq:bernstein}.

A breakthrough in the project of obtaining optimal guarantees was due to~\cite{Sha15}, who gave an analysis of Oja's algorithm that works when provided with a warm start: he showed that, when $k = 1$ and $\rank(\mtx A_i) = 1$ almost surely, Oja's algorithm converges in a number of steps logarithmic in $d$ if it is initialized in a neighborhood of the optimum, but his result does not extend to random initialization and it is unclear how to find a warm start in practice.
This restriction was lifted by~\cite{JaiJinKak16}, who were the first to show a global, efficient guarantee for Oja's algorithm when $k = 1$.
Subsequently, \cite{AllLi17} gave a global, efficient guarantee for Oja's algorithm in the $k > 1$ case, but under the restriction that $\rank(\mtx A_i) = 1$ almost surely.

The idea of analyzing Oja's algorithm by developing concentration bounds for products of random matrices was suggested by~\cite{HenWar20}, who also proved such non-asymptotic concentration bounds in a simplified setting.
Those bounds were later improved by~\cite{HuaNilTro20} who developed a different technique based on martingale inequalities for Schatten norms, following a strategy pursued by~\cite{JudNem08} and~\cite{naor2012banach} for other Banach space norms.
The concentration inequalities of~\cite{HuaNilTro20} are not sharp enough to recover optimal rates for Oja's algorithm on their own; in this work, we use a similar proof techniques to establish tailor-made concentration results for the Oja setting.

\subsection{Organization of the remainder of the paper}
In Section~\ref{sec:techniques}, we give our main results and an overview of our techniques.
Our main tool is a recursive inequality which proves a concentration result for the iterates of Oja's algorithm, which we state and prove in Section~\ref{sec:recursion}.

Our analysis of Oja's algorithm involves two distinct phases, which we analyze separately.
Since the argument for the second phase is simpler, we present it first in Section~\ref{sec:phase2}, and present the slightly more complicated argument for the first phase in Section~\ref{sec:phase1}.
We conclude in Section~\ref{sec:conclusion} with open questions and directions for future work.
The appendices contain omitted proofs and supplementary results for each section.

\subsection{Notation}
We write $\lambda_1 \geq \dots \geq \lambda_d$ for the eigenvalues of the symmetric matrix $\mtx M$, and we write $\gap := \lambda_k - \lambda_{k+1}$ for the gap between the $k$th and $(k+1)$th eigenvalue.
We write $\mtx V \in \RR^{d \times k}$ for the orthogonal matrix whose columns are the $k$ leading eigenvectors of $\mtx M$, and $\mtx U \in \RR^{d \times (d-k)}$ for the orthogonal matrix whose columns are the remaining eigenvectors.
Given an orthogonal matrix $\mtx W \in \RR^{d \times k}$, we write~\cite{DavKah70}
\begin{equation*}
\subspacedist{\mtx W}{\mtx V} = \|\mtx V \mtx V^* - \mtx W \mtx W^*\| = \|\mtx U^* \mtx W\|\,,
\end{equation*}

The symbol $\norm{\cdot}$ denotes the spectral norm (i.e., $\ell_2$ operator norm) of a matrix, which is equal to its maximum singular value.
For $p \geq 1$, the symbol $\norm{\cdot}_p$ denotes the Schatten $p$-norm,
which is the $\ell_p$ norm of the singular values of its argument.
We also define the $L_p$ norm of a random matrix $\mtx{X}$ 
as 
$$
\norm{\mtx X}_{p, p} := \big( \E \norm{\mtx{X}}_p^p \big)^{1/p}\,.
$$

We employ standard asymptotic notation $a = O(b)$ to indicate that $a \leq C b$ for a universal positive constant $C$, and write $a = \Theta(b)$ if $a = O(b)$ and $b = O(a)$.
The notations $\tilde O(\cdot)$ and $\tilde \Theta(\cdot)$ suppress polylogarithmic factors in the problem parameters.
When $t$ is a positive integer, we write $[t] := \{1, \dots, t\}$.
\section{Techniques and main results}\label{sec:techniques}

We focus throughout on the following setup:
\begin{assumption}\label{assume1}
The matrices $\mtx A_i$ are symmetric, independent, identically distributed samples from a distribution~$\lawa$, with expectation $\mtx M$.
\end{assumption}
Note that while we require that each $\mtx A_i$ is symmetric, we do not require that $\mtx A_i \succeq \zeromtx$.

The requirement that $\mtx A_i$ is symmetric is not as restrictive as it may seem, since we can replace $\mtx A_i$ by its \emph{Hermitian dilation}:
\begin{equation*}
\mathcal{D}(\mtx A_i) := \begin{pmatrix}
\zeromtx & \mtx A_i \\
\mtx A_i^* & \zeromtx
\end{pmatrix} \in \RR^{2d \times 2d}\,.
\end{equation*}
Estimating the leading eigenvectors of $\mathcal{D}(\mtx M)$ is equivalent to estimating the leading singular vectors of $\mtx M$.
Our results therefore extend to the non-symmetric streaming SVD problem as well.
We refer the reader to~\cite{Tro15:Introduction-Matrix} for more details about this standard reduction.

The second requirement establishes that the random errors are bounded in a suitable norm.
We write $\stie{d}{k}$ for the Stiefel manifold of $d \times k$ matrices with orthonormal columns.
\begin{assumption}\label{assume2}
If $\mtx A \sim \lawa$, then $\sup_{\mtx P \in \stie{d}{k}} \|\mtx P^* (\mtx A - \mtx M)\|_2 \leq M$ almost surely.
\end{assumption}

Note that for any matrix $\mtx{X}\in \mathbb{R}^{d\times d}$, 
\[\sup_{P \in \stie{d}{k}} \|\mtx P^*\mtx{X}\|_2 = \left(\sum_{i=1}^k\sigma_i(\mtx{X})^2\right)^{1/2},\quad 1\leq k\leq d,\]
where $\sigma_1(\mtx{X}) \geq \sigma_2(\mtx{X})\geq \cdots \geq \sigma_d(\mtx{X})$ are the singular values of $\mtx{X}$.
This norm, sometimes known as the $(2, k)$ norm~\cite{LiTsi88} or the Ky Fan $2$-$k$ norm~\cite{DoaVav16}, satisfies
\begin{equation*}
\|\mtx X\| \leq \sup_{P \in \stie{d}{k}} \|\mtx P^*\mtx{X}\|_2 \leq \sqrt k \|\mtx X\| \leq \|\mtx X\|_2\,.
\end{equation*}
This choice of norm generalizes the error assumptions in the literature.
In the $k = 1$ case, it agrees with the operator norm, which is the condition used by \cite{JaiJinKak16}; and it weakens the requirement of~\cite{AllLi17} that $\|\mtx A_i\|_2 \leq 1$ almost surely.

The following theorem summarizes our main results for Oja's algorithm.

\begin{theorem}[Main, informal]\label{thm:main}
Adopt Assumptions~\ref{assume1} and~\ref{assume2}.
Let $\lambda_1 \geq \dots \lambda_d$ be the eigenvalues of $\mtx M$, and let $\gap = \lambda_k - \lambda_{k+1}$.

For every $\delta\in (0,1)$, define learning rates
\begin{equation*}
T_0 = \tilde \Theta\left(\frac{k M^2}{\delta^2 \rho_k^2}\right)\,, \quad \beta = \tilde \Theta\left(\frac{M^2}{\rho_k^2}\right)\,, \quad \eta_t=\left\{\begin{array}{ll}
\tilde \Theta\left(\frac{1}{\rho_kT_0}\right),& t\leq T_0\\
\Theta \left(\frac{1}{\rho_k (\beta + t-T_0)}\right),& t>T_0.
\end{array} \right.
\end{equation*}

Let $\mtx{V} \in \RR^{d \times k}$ be the orthogonal matrix whose columns are the $k$ leading eigenvectors of $\mtx M$. Then for any $T > T_0$, the output $\mtx{Q}_T$ of Oja's algorithm satisfies
\[\subspacedist{\mtx Q_T}{\mtx V}\leq C' \frac{M}{\rho_k} \sqrt{\frac{\log(M k/\rho_k \delta)}{T - T_0}}\]
with probability at least $1-\delta$, where $C'$ is a universal positive constant.
\end{theorem}

To prove Theorem~\ref{thm:main}, we adopt a two-phase analysis.
Our first result shows that after $T_0$ iterations, the output of Oja's algorithm satisfies $\|\mtx U^* \mtx Q_{T_0}(\mtx V^* \mtx Q_{T_0})^{-1}\| \leq 1$ with high probability.

\begin{theorem}[Phase I, informal]\label{thm:phase1_informal}
Adopt the same setting as Theorem~\ref{thm:main}, and let $\mtx Z_0 \in \RR^{d \times k}$ have i.i.d.~Gaussian entries.
Let
\begin{equation*}
T_0 = \Theta\left(\frac{k M^2}{\delta^2 \rho_k^2}\big(\log(dM/\delta\rho_k)\big)^4\right)\,.
\end{equation*}
Then after $T_0$ iterations of Oja's algorithm with constant step size $\eta = \Theta\left(\frac{\log(d/\delta)}{\rho_kT_0}\right)$ and initialization $\mtx Z_0$, the output $\mtx Q_{T_0}$ satisfies
\begin{equation*}
\|\mtx U^* \mtx Q_{T_0} (\mtx V^* \mtx Q_{T_0})^{-1}\| \leq 1
\end{equation*}
with probability at least $1-\delta$.
\end{theorem}

Our analysis of the second phase shows that, if Oja's algorithm is initialized with \emph{any} matrix satisfying $\|\mtx U^* \mtx Q_{0}(\mtx V^* \mtx Q_{0})^{-1}\| \leq 1$, then the output of Oja's algorithm decays at the rate $O(1/\sqrt T)$.

\begin{theorem}[Phase II, informal]\label{thm:phase2_informal}
Adopt the same setting as Theorem~\ref{thm:main}, and suppose that $\mtx Z_0 \in \RR^{d \times k}$ satisfies $\|\mtx U^* \mtx Z_{0} (\mtx V^* \mtx Z_{0})^{-1}\| \leq 1$.
Then after $T$ iterations of Oja's algorithm with step size $\eta_i = \frac{8}{(\beta + i)\rho_k}$ with $\beta = \Theta\left(\frac{M^2}{\rho_k^2} \log \left(\frac{M k }{\rho_k \delta} \right)\right)$ and initialization $\mtx Q_0$, the output $\mtx Q_T$ satisfies
\begin{equation}\label{eq:phase2_error}
\subspacedist{\mtx Q_T}{\mtx V} \leq 2 \econst \sqrt{\frac{\beta + 1}{\beta + T}}
\end{equation}
with probability at least $1-\delta$.
\end{theorem}
This error guarantee is completely dimension free, and depends only logarithmically on $k$ and the failure probability $\delta$.

Theorem~\ref{thm:main} follows directly from Theorems~\ref{thm:phase1_informal} and~\ref{thm:phase2_informal}.
Theorem~\ref{thm:phase1_informal} guarantees that with probability $1-\delta$, the output of Phase I is a suitable initialization for Phase II, and, conditioned on this good event, Theorem~\ref{thm:phase2_informal} guarantees that the output of the second phase has error $O(\sqrt{\beta/T})$ with probability $1-\delta$.
By concatenating the analysis of the two phases and using the union bound, we obtain that the resulting two-phase algorithm succeeds with probability at least $1 - 2 \delta$, yielding Theorem~\ref{thm:main}.

In the remainder of this section, we describe the main technical tools we employ in our argument.

\subsection{A recursive expression}
To simplify the argument, we recall the following result of \cite[Lemma 2.2]{AllLi17}:
\begin{lemma}\label{lem:dist_to_w}
For all $t \geq 0$,
\begin{equation*}
\subspacedist{\mtx Q_t}{\mtx V} = \|\mtx U^* \mtx Q_t\| \leq \|\mtx U^* \mtx Q_t (\mtx V^* \mtx Q_t)^{-1}\| = \|\mtx U^* \mtx Z_t (\mtx V^* \mtx Z_t)^{-1}\|\,.
\end{equation*}
\end{lemma}
We therefore focus on bounding the norm of the matrix
\begin{equation}\label{eq:w_def}
\mtx W_t := \mtx U^* \mtx Z_t (\mtx V^* \mtx Z_t)^{-1}\,.
\end{equation}

Under the assumption that $\eta_t$ is small, we might expect that we can write $\mtx W_t$ as a sum of the dominant term
\begin{align}
\mtx H_t & := \mtx U^* (\Id + \eta_t \mtx M) \mtx Z_{t-1} (\mtx V^* (\Id + \eta_t \mtx M) \mtx Z_{t-1})^{-1} \label{eq:h_def}
\end{align}
plus lower order terms.

To argue that $\mtx W_t$ is close to $\mtx H_t$, we need to argue that the inverse $(\mtx V^* \mtx Z_t)^{-1}$ does not blow up, which will be the case so long as the fluctuation term $\eta_t \mtx V^* (\mtx A_t - \mtx M)\mtx Z_{t-1}$ is smaller than the main term $\mtx V^* (\Id + \eta_t \mtx M) \mtx Z_{t-1}$.
In order to make this requirement precise, we write
\begin{equation}\label{eq:delta_def}
\mtx \Delta_t := \eta \mtx V^* (\mtx A_t - \mtx M)\mtx Z_{t-1} (\mtx V^*(\Id + \eta_t \mtx M) \mtx Z_{t-1})^{-1}\,.
\end{equation}
So long as this matrix has small norm, the inverse term will be well behaved.
As we discuss in the following section, we will be able to guarantee that this is the case by conditioning on an appropriate good event.

The following lemma shows that, modulo a term involving $\mtx \Delta_t$, we can indeed express $\mtx W_t$ as $\mtx H_t$ plus a small correction.
\begin{lemma}\label{lem:decomposition}
Let $\mtx W_t$, $\mtx H_t$, and $\mtx \Delta_t$ be defined as in~\eqref{eq:w_def}--\eqref{eq:delta_def}.
Then we can write
\begin{equation}\label{eq:decomposition}
\mtx{W}_t(\Id - {\mtx{\Delta}}_t^2) = \mtx{H}_t + \mtx{J}_{t,1} + \mtx{J}_{t,2}\,,
\end{equation}
for matrices $\mtx{J}_{t,1}$ and $\mtx{J}_{t,2}$ of norm $O(\eta_t)$ and $O(\eta_t^2)$, respectively.
\end{lemma}

Below, in Propositions~\ref{j_bounds} and~\ref{prop:one_step_recurrence}, we use Lemma~\ref{lem:decomposition} to develop an explicit recursive bound on the norm of $\mtx W_t$.

\subsection{Matrix concentration via smoothness}\label{sec:smoothness}
In order to exploit the expression~\eqref{eq:decomposition}, we need concentration inequalities that allow us to conclude that $\mtx W_t$ is near $\mtx H_t$ with high probability.
\cite{HuaNilTro20} recently developed new tools to control the norms of products of independent random matrices, in an attempt to extend the mature toolset for bounding \emph{sums} of random matrices to the product setting.
Their techniques are based on a simple but deep property of the Schatten $p$-norms known as \emph{uniform smoothness}.
The most elementary expression of this fact is the following inequality, which is the analogue of~\eqref{eq:pythag} for the $L_p$ norm.
\begin{proposition}[{\cite[Proposition 4.3]{HuaNilTro20}}]\label{prop:smooth}
Let $\mtx X$ and $\mtx Y$ be random matrices of the same size, with $\E[\mtx Y \mid \mtx X] = \zeromtx$.
Then for any $p \geq 2$,
\begin{equation*}
\|\mtx X + \mtx Y + \mtx Z\|_{p, p}^2 \leq \|\mtx X\|_{p, p}^2 + (p-1) \|\mtx Y\|_{p, p}^2\,.
\end{equation*}
\end{proposition}

We will employ the following corollary of Proposition~\ref{prop:smooth}, which extends the inequality to non-centered random matrices.
\begin{proposition}\label{non-centered-smoothness}
Let $\mtx X$, $\mtx Y$, and $\mtx Z$ be random matrices of the same size, with $\E[\mtx Y \mid \mtx X] = \zeromtx$.
Then for any $p \geq 2$ and $\lambda > 0$,
\begin{equation*}
\|\mtx X + \mtx Y + \mtx Z\|_{p, p}^2 \leq (1 + \lambda)(\|\mtx X\|_{p, p}^2 + (p-1) \|\mtx Y\|_{p, p}^2 + \lambda^{-1} \|\mtx Z\|_{p, q}^2)
\end{equation*}
\end{proposition}
The benefit of working in the $L_p$ norm is that bounding this norm for $p$ large yields good tail bounds on the operator norm, which are not available if the argument is carried out solely in expected Frobenius norm.
We will rely heavily on this fact heavily in our argument.

\subsection{Conditioning on good events}
Obtaining control on $\mtx W_t$ via~\eqref{eq:decomposition} requires ensuring that the matrix $\Id - \mtx \Delta_t^2$ is invertible, with inverse of bounded norm.
To accomplish this, we define a sequence of good events $\good_0 \supset \good_1 \supset \dots$, where each $\good_i$ is measurable with respect to the $\sigma$-algebra $\filtration_i := \sigma(\mtx Z_0, \mtx Y_1, \dots, \mtx Y_i)$.
We write $\bone_i$ for the indicator of the event $\good_i$, and we will define $\good_i$ in such a way that $(\Id - \mtx \Delta_t^2\bone_{t-1})$ is invertible almost surely.

During Phase II, the good events are defined by
\begin{align*}
\good_0 & := \{\|\mtx W_0\| \leq 1\} \\
\good_i & := \{\|\mtx W_i\| \leq \gamma\} \cap \good_{i-1}\,,  \quad \forall i \geq 1
\end{align*}
for some $\gamma \geq 1$ to be specified.
Since Assumption~\ref{assume2} implies that $\|\mtx A_i - \mtx M\| \leq M$ almost surely, this definition guarantees that for all $i \geq 1$,
\begin{equation}\label{eq:conditional_requirement}
\|\mtx V^* (\mtx A_i - \mtx M) \mtx U \mtx W_{i-1} \bone_{i-1}\| \leq M \gamma \quad \text{almost surely.}
\end{equation}

As we show in Proposition~\ref{j_bounds} below, if the step size is sufficiently small, then~\eqref{eq:conditional_requirement} implies that $\Id - \mtx \Delta_t^2$ is almost surely invertible on $\good_{t-1}$, which allows us to employ~\eqref{eq:decomposition} to bound the norm of $\mtx W_t \bone_{t-1}$.

During Phase I, we condition on a slightly more complicated set of events, which we describe explicitly in Section~\ref{sec:phase1}.
However, these events are constructed so that~\eqref{eq:conditional_requirement} still holds for all $i \geq 1$.

Our matrix concentration results described in Section~\ref{sec:smoothness} allow us to show that, during both Phase I and Phase II, $\|\mtx W_t \bone_{t-1}\|$ is small with high probability, for all $t \geq 1$.
Using this fact, we show that, conditioned on $\good_{t-1}$, the probability that $\good_t$ holds is also large.
Bounding the failure probability at each step, we are able to conclude that, conditioned on the initialization event $\good_0$, the good events $\good_t$ hold for all $t \geq 1$ with high probability.

\section{Main recursive bound}\label{sec:recursion}
In this section, we state our main recursive bound, which we use in both Phase I and Phase II.
A proof appears in Section~\ref{main-recurrence-proof}.
\begin{theorem}\label{thm:main_recurrence}
Let $t$ be a positive integer, and for all $i \in [t]$, let $\ep_i = 2 \eta_i M(1+\gamma)$.
Let $\bone_1, \dots, \bone_t$ be the indicator functions of a sequence of good events satisfying~\eqref{eq:conditional_requirement} for all $i \in [t]$.

Assume that for all $i \in [t]$,
\begin{equation}\label{main_assumptions}
\ep_i  \leq \frac 12 \,, \quad \quad 
\eta_i \|\mtx M\|  \leq \frac 12 \,,  \quad \quad
\econst^{- \eta_i \rho_k/4}  \leq \frac{\ep_i}{\ep_{i-1}}\,,
\end{equation}
with the convention that the last requirement is vacuous when $i = 1$.
Then for any $p \geq 2$,
\begin{equation}\label{eq:main_bound1}
\|\mtx W_t \bone_t \|_{p, p}^2 \leq \|\mtx W_t \bone_{t-1}\|_{p, p}^2 \leq \econst^{- s_t \rho_k} \|\mtx W_0 \bone_0\|_{p, p}^2 + C_1 p  \ep_t^2 \sum_{i=0}^{t-1} \|\mtx W_i \bone_i\|_{p, p}^2 + C_2 p k^{2/p} \ep_t^2 t\,,
\end{equation}
where $s_t = \sum_{i=1}^t \eta_i$, $C_1 = 21$, and $C_2 = 5$.
Moreover, if in addition for all $i \in [t]$,
\begin{align}
p \ep_i^2 & \leq \frac{\eta_i \rho_k}{50}\,, \label{assume:p_small}
\end{align}
then
\begin{equation}\label{eq:main_bound2}
\|\mtx W_t \bone_t \|_{p, p}^2 \leq \|\mtx W_t \bone_{t-1}\|_{p, p}^2 \leq \econst^{- s_t \rho_k/2} \|\mtx W_0 \bone_0\|_{p, p}^2 + C_2 p k^{2/p} \ep_t^2 t\,.
\end{equation}

\end{theorem}
Theorem~\ref{main_assumptions} shows that, up to small error, $\|\mtx W_t \bone_{t-1}\|_{p, p}^2$ decays exponentially fast.
We will use this fact to prove high probability bounds on $\|\mtx W_t \bone_{t-1}\|$, which then imply bounds on $\|\mtx W_t\|$.
\section{Phase II}\label{sec:phase2}
In this section, we use Theorem~\ref{thm:main_recurrence} to prove a formal version of Theorem~\ref{thm:phase2_informal}.

For this phase, recall that we define the good events $\good_i$ by
\begin{equation}\label{eq:phase2_good}
\good_0 = \{\|\mtx W_0\| \leq 1\}\,, \quad \quad 
\good_i = \{\|\mtx W_i\| \leq \gamma\} \cap \good_{i-1}\,, \quad \forall i \geq 1\,.
\end{equation}
For Phase II, we set $\gamma = \sqrt{2} \econst$.

We first show that, with a specific step-size schedule, we obtain good bounds on the norm of the last iterate.
\begin{proposition}\label{prop:phase2_norm_bound}
Define the good events as in~\eqref{eq:phase2_good}.
Set $\eta_i = \frac{\alpha}{(\beta + i) \rho_k}$, for positive quantities $\alpha$ and $\beta$, and define the normalized gap
\begin{equation}\label{eq:barrho}
\barrho_k = \min\left\{\frac{\rho_k}{M}, \frac{\rho_k}{\|\mtx M\|}, 1\right\}\,.
\end{equation}
If
\begin{equation}\label{alpha_beta_conds}
\alpha \geq 8\,, \quad \beta \geq \frac{4 (1+\sqrt{2}\econst) \alpha}{\barrho_k}\,,
\end{equation}
then for any $t \geq 1$,
\begin{equation}\label{eq:pq-normbound_decay}
\|\mtx W_t \bone_t\|_{p, p}^2 \leq k^{2/p} \left( \frac{\beta + 1}{\beta + t}\right)^ \alpha + p k^{2/p} \cdot \left(\frac{C_3 \alpha}{\barrho_k}\right)^2 \cdot  \frac{t}{(\beta + t)^2}\,,
\end{equation}
where $C_3$ is a numerical constant less than $175$.
\end{proposition}
\begin{proof}
Since the good events defined in~\eqref{eq:phase2_good} satisfy~\eqref{eq:conditional_requirement}, we can apply Theorem~\ref{thm:main_recurrence}.
In the appendix, we show (Lemma~\ref{phase2-assumption-verification}) that \eqref{alpha_beta_conds} implies that the assumptions in~\eqref{main_assumptions} hold.
Theorem~\ref{thm:main_recurrence} then yields
\begin{align*}
\|\mtx W_t \bone_t\|_{p, p}^2 & \leq \econst^{-s_t \rho_k} \|\mtx W_0 \bone_0\|_{p, p}^2 + C_1 p \ep_t^2 \sum_{i=1}^{t-1} \|\mtx W_i \bone_i\|_{p, p}^2 + C_2 p k^{2/p} \ep_t^2 t\\
& \leq \econst^{-s_t \rho_k} k^{2/p} + (C_1\gamma^2 + C_2) p  k^{2/p} \ep_t^2 t\,,
\end{align*}
since~\eqref{eq:phase2_good} implies $\|\mtx W_0 \bone_0\|_{p, p}^2 \leq k^{2/p}$ and $\|\mtx W_i \bone_i\|_{p, p}^2 \leq \gamma^2 k^{2/p}$ for all $i \geq 1$.

The definition of $\eta_i$ implies
\[\rho_k s_t = \alpha\sum_{i=1}^t\frac{1}{\beta+i}\geq \alpha \log\left(\frac{\beta+t}{\beta+1}\right).\]
We obtain
\begin{equation*}
\|\mtx W_t \bone_t\|_{p, p}^2 \leq k^{2/p} \left( \frac{\beta + 1}{\beta + t}\right)^ \alpha + p k^{2/p} \cdot \left(\frac{C_3 \alpha}{\barrho_k}\right)^2 \cdot  \frac{t}{(\beta + t)^2}\,,
\end{equation*}
where
\begin{equation*}
C_3 = (C_1 \gamma^2 + C_2)^{1/2} C_\ep < 175\,,
\end{equation*}
as desired.
\end{proof}

Finally, we remove the conditioning and prove the full version of Theorem~\ref{thm:phase2_informal}.

\begin{theorem}\label{thm:phase2}
Assume $\|\mtx W_0\| \leq 1$, and adopt the step size $\eta_i = \frac{\alpha}{(\beta + i) \rho_k}$,
with
\begin{equation*}
\alpha \geq 8\,, \quad \beta \geq  2 \left(\frac{C_3 \alpha}{\barrho_k}\right)^2 \log\left(\frac{C_3 \alpha}{\barrho_k} \cdot 2k/\delta\right)\,,
\end{equation*}
where $\bar \rho_k$ is as in~\eqref{eq:barrho} and $C_3$ is as in~\eqref{eq:pq-normbound_decay}.
Then
\begin{equation*}
\|\mtx W_T\| \leq 2 \econst \sqrt{\frac{\beta + 1}{\beta + T}}
\end{equation*}
with probability at least $1 - \delta$.
\end{theorem}
\begin{proof}
For any $s\geq 0$, it holds $\Prob{\|\mtx{W}_T\| \geq s} \leq \Prob{\|\mtx{W}_T \mathbb{1}_T\| \geq s} + \Prob{\good_T^C}$.
First, we have
\begin{equation*}
\Prob{\mathcal G_T^C} \leq \Prob{\good_0^C} + \sum_{j=1}^T \Prob{\mathcal G_j^C \cap \mathcal G_{j-1}}\,.
\end{equation*}
Since we have assumed that the initialization satisfies $\|\mtx W_0\| \leq 1$, the event $\good_0$ holds with probability $1$, so it suffices to bound the second term.
By Markov's inequality, we have
\begin{equation*}
\Prob{\mathcal G_j^C \cap \mathcal G_{j-1}} = \Prob{\|\mtx W_j \bone_{j-1}\| \geq \gamma} \leq \inf_{p \geq 2} \gamma^{-p} \|\mtx W_j \bone_{j-1}\|_{p, p}^p\,.
\end{equation*}
For fixed $j\geq 1$, we choose $p =(\beta + j) \cdot \frac{\barrho_k^2}{C_3^2 \alpha^2 }$. It follows from \eqref{eq:pq-normbound_decay} that, 
\begin{align*}
\gamma^{-p}\|\mtx W_j \mathbb 1_{j-1}\|_{p,p}^p &\leq \left(\frac{1}{\gamma^2}k^{2/p}\left(\frac{\beta+1}{\beta+j}\right)^{\alpha} +  \frac{1}{\gamma^2}p k^{2/p} \cdot \frac{C_3^2 \alpha^2 }{\barrho_k^2} \cdot  \frac{j}{(\beta + j)^2}\right)^{p/2}\\
&\leq k\left(\frac{1}{2\econst^2} + \frac{1}{2 \econst^2} \frac{j}{\beta +j}\right)^{p/2}\\
& \leq k\econst^{-p} = k\exp\left(-(\beta + j) \cdot \frac{\barrho_k^2}{C_3^2 \alpha^2}\right)\,. 
\end{align*}
Therefore, for any $T\geq 1$,
\[\sum_{j = 1}^{T} \Prob{\mathcal G_j^C | \mathcal G_{j-1}}\leq k\sum_{j = 1}^{T} \exp\left(-(\beta + j) \cdot \frac{\barrho_k^2}{C_3^2 \alpha^2}\right)\leq k \frac{C_3^2 \alpha^2}{\barrho_k^2} \econst^{-\beta \cdot \frac{\barrho_k^2}{C_3^2 \alpha^2}}.\]
This quantity is smaller than $\delta/2$ if
\[\beta \geq 2 \frac{C_3^2 \alpha^2}{\barrho_k^2} \log\left(\frac{C_3 \alpha M}{\barrho_k} \cdot 2k/\delta\right)\,.\]

It remains to bound $\Prob{\|\mtx{W}_T \mathbb{1}_T\| \geq s}$.
A simple argument (Lemma~\ref{lem:phase1-last-bound}) based on~\eqref{eq:pq-normbound_decay} shows that this probability is at least $\delta/2$ for
\begin{equation*}
s = 2 \econst \sqrt{\frac{\beta + 1}{\beta + T}}\,.
\end{equation*} 
The claim follows.
\end{proof}

\section{Phase I}\label{sec:phase1}
In this section, we describe the slightly more delicate proof of the formal version of Theorem~\ref{thm:phase1_informal}.
As in Section~\ref{sec:phase2}, we will employ Theorem~\ref{thm:main_recurrence}.
However, we will also need to develop an auxiliary recurrence to bound the growth of an additional matrix sequence.

Before we analyze Phase I, we first show that we can reduce to the case that that $\lawa$ has finite support.
We prove the following result in Appendix~\ref{ghost}.
\begin{proposition}\label{prop:reduction}
Fix $\rho > 0$.
Suppose that there exists a choice of constant step size $\eta$ and $T_0 \geq \frac{9 M}{\rho \delta} \log(d/\delta)$ such that for any \emph{finitely-supported} distribution with support size at most $T_0^3$ satisfying Assumptions~\ref{assume1} and~\ref{assume2} and with $\rho_k \geq \rho/2$, we have
\begin{equation}\label{better_error}
\|{\mtx U}^* \mtx Q_{T_0} ({\mtx V}^* \mtx Q_{T_0})^{-1}\| \leq \frac 16
\end{equation}
with probability at least $1 - \delta/3$.

Then for this same $\eta$ and $T_0$ it in fact holds that for \emph{any} distribution satisfying Assumptions~\ref{assume1} and~\ref{assume2} and with $\rho_k \geq \rho$, we have
\begin{equation*}
\|{\mtx U}^* \mtx Q_{T_0} ({\mtx V}^* \mtx Q_{T_0})^{-1}\| \leq 1
\end{equation*}
with probability at least $1 - \delta$.
\end{proposition}

Proposition~\ref{prop:reduction} implies that it suffices to prove the error guarantee~\eqref{better_error} in the special case when $\lawa$ has finite support of cardinality at most $T_0^3$.

Let us fix a time horizon $T_0$ and assume in what follows that $m := |\mathrm{supp}(\lawa)| \leq T_0^3$.We begin by defining the good events for Phase I.
We adopt a constant step size $\eta$, to be specified.
Denote 
\[\mathcal E := \{M^{-1}(\mtx A - \mtx M) \mtx U \mtx U^*: \mtx A \in \mathrm{supp}(\lawa)\}.\]
For $i \geq 1$, we will set
\begin{equation}\label{eq:phase1_gooda}
\good_i = \{\max_{\mtx E \in \mathcal E} \|\mtx V^* \mtx E \mtx U \mtx W_i\| \leq \gamma\} \cap \good_{i-1}\,.
\end{equation}
Note that this choice satisfies~\eqref{eq:conditional_requirement} for all $i > 1$.

To define the initial good event $\good_0$, we need to define a larger set of matrices to condition on.
For all $r, \ell \geq 1$, set
\begin{align*}
\mathcal E_{r, \ell} := \{\mtx V^* \mtx F_1 \cdots \mtx F_r \mtx U : &\mtx F_i \in \mathcal E \text{ for at most $\ell$ distinct indices $i \in [r]$}, \\
&\text{and $\mtx F_i = (1 + \eta \lambda_{k+1})^{-1} (\Id + \eta \mtx M)\mtx U \mtx U^* $ otherwise}\}
\end{align*}
The set $\mathcal E_{r, \ell}$ has cardinality less than $(r(m+1))^\ell$, and $\|\mtx E\|_2 \leq 1$ for any $\mtx E \in \mathcal E_{r, \ell}$, and any $r, \ell \geq 1$.
We have defined $\mathcal E_{r, \ell}$ so that control over $\max_{\mtx E \in \mathcal E_{r +1, \ell+1}} \|\mtx E \mtx W_{t-1}\|$ gives control over $\max_{\mtx E \in \mathcal E_{r , \ell}} \|\mtx E \mtx W_t\|$.

Finally, we define
\begin{equation}\label{eq:phase1_goodb}
\mathcal G_0 := \bigcap_{r, \ell = 1}^{T_0+1} \left\{\max_{\mtx E \in \mathcal E_{r, \ell }} \|\mtx E \mtx W_0\|_2 \leq \frac{\sqrt{\ell} \gamma}{\sqrt{2}\econst} \right\} \cap \{\|\mtx W_0\|_2 \leq \sqrt{d}\gamma\}\,. 
\end{equation}
Since $\mtx V^* (\mtx A_1 - \mtx M) \mtx U \in \mathcal E_{1, 1}$ almost surely, this choice satisfies \eqref{eq:conditional_requirement} for $i = 1.$

Our strategy will be similar to the one used in Section~\ref{sec:phase2}.
However, in order to show that the good events $\good_i$ hold with high probability, we will also need a second recurrence that allows us to control the norm of matrices of the form $\mtx E \mtx W_t$, for $\mtx E \in \mathcal E_{r, \ell}$.
The details appear in Section~\ref{sec:additional-phase1}.

\section{Conclusion}\label{sec:conclusion}
This work gives the first nearly optimal analysis of Oja's algorithm for streaming PCA beyond the rank one case.
Our analysis is conceptually simple: we show that the spectral norm of the matrix $\mtx W_t$ concentrates well around its expectation, once we condition on $\mtx W_{t-1}$ having the same behavior.
And our concentration results are strong enough that we can pay to union bound over the entire course of the algorithm, to show that $\mtx W_t$ is well behaved for all $t \geq 1$.

The matrix concentration techniques we have applied here could be useful in analyzing other PCA-like algorithms, or, more generally, other stochastic algorithms for simple non-convex optimization problems.
An interesting question is whether these techniques can prove \emph{gap-free} rates for Oja's algorithm outside the rank-one setting.
This would extend the results of~\cite{AllLi17} to the general case.

Finally, we stress that the algorithm we have described here requires \emph{a priori} knowledge of the problem parameters (including the gap $\rho_k$) to set the step sizes, which is a serious limitation in practice.
Recently, \cite{HenWar19} developed a data-driven procedure to adaptively select the optimal step sizes.
Obtaining theoretical guarantees for this or similar algorithms is an important open problem.

\section*{Acknowledgement}
We thank Joel Tropp and Amelia Henriksen for valuable discussions which greatly improved this manuscript.


\appendix
\section{Additional results for Section~\ref{sec:recursion}}
The following proposition develops the expansion described in Lemma~\ref{lem:decomposition} and gives explicit bounds on the norms of the error matrices $\mtx J_{t, 1}$ and $\mtx J_{t, 2}$.

We recall the following definitions
\begin{align*}
\mtx W_t & = \mtx U^* \mtx Z_t (\mtx V^* \mtx Z_t)^{-1} \\
\mtx H_t & = \mtx U^* (\Id + \eta \mtx M) \mtx Z_{t-1} (\mtx V^* (\Id + \eta \mtx M) \mtx Z_{t-1})^{-1} \\
\mtx \Delta_t & = \eta_t \mtx V^*(\mtx A_t - \mtx M) \mtx Z_{t-1} (\mtx V^*(\Id + \eta_t \mtx M) \mtx Z_{t-1})^{-1}
\end{align*}

\begin{proposition}\label{j_bounds}
Let $t \geq 1$.
Assume that $\eta_t$ is small enough that $\mtx M \succeq - \frac{1}{2 \eta_t} \Id$, and assume that~\eqref{eq:conditional_requirement} holds for $i = t$.
Let
\begin{align*}
E_t & = (k^{1/p} + 2 \|\mtx W_{t-1} \bone_{t-1}\|_{p, p}) \\
\ep_t & = 2 \eta_t M(1 + \gamma)\,.
\end{align*}

Then $\|\mtx \Delta_t \bone_{t-1}\| \leq \ep_t$ almost surely, and
\begin{equation}\label{eq:w_decomp}
\mtx{W}_t(\Id - {\mtx{\Delta}}_t^2) = \mtx{H}_t + \mtx{J}_{t,1} + \mtx{J}_{t,2}
\end{equation}
for $\mtx J_{t, 1}$ and $\mtx J_{t, 2}$ satisfying
\begin{align*}
\|\mtx{J}_{t,1} \bone_{t-1}\|_{p, p} & \leq  E_t \ep_t  \\ 
\|\mtx{J}_{t,2} \bone_{t-1}\|_{p, p} & \leq E_t \ep_t^2\,,  \\ 
\end{align*}
and $\E[\mtx J_{t, 1} \mid \filtration_{t-1}] = \zeromtx$.
\end{proposition}
\begin{proof}
We employ the notation of the proof of Lemma~\ref{lem:decomposition}.
(See Appendix~\ref{sec:omit}.)
First, we show the bound on $\mtx \Delta_t$.
Since $\eta_t \mtx M \succeq - \frac{1}{2} \Id$, we have $\|\mtx V^* (\Id + \eta_t \mtx M)^{-1} \mtx V\| \leq 2$. 
Moreover, since $\|\mtx{V}^* (\mtx{A}_t - \mtx{M})\mtx{U} W_{t-1}\| \leq M \gamma$ almost surely, we have that 
\begin{align*}
\|{\mtx{\Delta}}_t\mathbb{1}_{t-1}\| &\leq 2 \|\eta_t\mtx{V}^*(\mtx{A}_t-\mtx{M})(\mtx{U} \mtx{U}^* + \mtx{V} \mtx{V}^*)\mtx Z_{t-1} (\mtx{V}^*\mtx{Z}_{t-1})^{-1}\mathbb{1}_{t-1}\|\\
&\leq 2\eta_t\|\mtx{V}^*(\mtx{A}_t-\mtx{M}) \mtx{U}\mtx{U}^* \mtx Z_{t-1} (\mtx{V}^*\mtx{Z}_{t-1})^{-1}\mathbb{1}_{t-1}\| + 2\eta_t\|\mtx{V}^*(\mtx{A}_t-\mtx{M}) \mtx{V} \|\\
&=  2\eta_t\|\mtx{V}^*(\mtx{A}_t-\mtx{M}) \mtx{U}\mtx{W}_{t-1}\mathbb{1}_{t-1}\| + 2\eta_t\|\mtx{V}^*(\mtx{A}_t-\mtx{M}) \mtx{V} \|\\
&\leq 2\eta_t M(1+\gamma) =: \eps_t\,.
\end{align*}

We can bound $\|\widehat{\mtx \Delta}_t \bone_{t-1} \|_{p, p}$ by a similar argument.
First, note that Assumption~\ref{assume2} implies that $\|\mtx A_t - \mtx M\| \leq M$ almost surely.
Hence
\begin{align*}
\|\widehat{\mtx \Delta}_t \mathbb{1}_{t-1}\|_{p,p} & \leq 2\eta_t\|\mtx{U}^*(\mtx{A}_t-\mtx{M}) \mtx{U}\mtx{U}^* Z_{t-1} (\mtx{V}^*\mtx{Z}_{t-1})^{-1}\mathbb{1}_{t-1}\|_{p,p} +2\eta_t\|\mtx{U}^*(\mtx{A}_t-\mtx{M}) \mtx{V} \mathbb{1}_{t-1}\|_{p,p} \\
&= 2\eta_t\|\mtx{U}^*(\mtx{A}_t-\mtx{M}) \mtx{U}\|\, \|\mtx{W}_{t-1}\mathbb{1}_{t-1}\|_{p,p} + 2\eta_t\|\mtx{U}^*(\mtx{A}_t-\mtx{M}) \mtx{V} \mathbb{1}_{t-1}\|_{p,p} \\
& \leq( \|\mtx{W}_{t-1}\mathbb{1}_{t-1}\|_{p,p}+ k^{1/p})2\eta_t  M\\
& \leq ( \|\mtx{W}_{t-1}\mathbb{1}_{t-1}\|_{p,p}+ k^{1/p}) \eps_t\,,
\end{align*}

Finally, we have
\[\|\mtx{H}_t\mathbb{1}_{t-1}\|_{p,p} \leq \frac{1+\eta_t\lambda_{k+1}}{1+\eta_t\lambda_k}\|\mtx{W}_{t-1}\mathbb{1}_{t-1}\|_{p,p}\leq \|\mtx{W}_{t-1}\mathbb{1}_{t-1}\|_{p,p}\,.\]

We now employ Lemma~\ref{lem:decomposition}.
The term $\mtx J_{t, 1}$ satisfies
\begin{equation*}
\E[\mtx J_{t, 1} \bone_{t-1} | \filtration_{t-1}] = \zeromtx\,, 
\end{equation*}
and we have
\begin{align}
\|\mtx J_{t,1} \bone_{t-1}\|_{p, p} & \leq \|\widehat{\mtx \Delta}_t \bone_{t-1} \|_{p, p} + \|\mtx H_t \bone_{t-1}\|_{p, p}\|\mtx \Delta_t \bone_{t-1}\| \\
& \leq ( \|\mtx{W}_{t-1}\mathbb{1}_{t-1}\|_{p,p}+ k^{1/p}) \eps_t + \|\mtx{W}_{t-1}\mathbb{1}_{t-1}\|_{p,p} \eps_t \\
& \leq E_t \eps_t\,.
\end{align}
Finally,
\begin{equation*}
\|\mtx J_{t, 2}\|_{p, p} \leq \|\widehat{\mtx \Delta}_t \bone_{t-1}\|_{p, p}\|\mtx \Delta_{t} \bone_{t-1}\| \leq ( \|\mtx{W}_{t-1}\mathbb{1}_{t-1}\|_{p,p}+ k^{1/p}) \eps_t^2 \leq E_t \eps_t^2\,.
\end{equation*}
\end{proof}

Combining Proposition~\ref{j_bounds} with Proposition~\ref{non-centered-smoothness} immediately yields a recursive bound.
\begin{proposition}\label{prop:one_step_recurrence}
Adopt the setting of Proposition~\ref{j_bounds}.
If $\ep_t \leq 1/2$, then
\begin{equation}\label{eq:first_recursive}
\|\mtx{W}_t\mathbb{1}_t\|_{p,p}^2\leq \|\mtx{W}_t\mathbb{1}_{t-1}\|_{p,p}^2\leq K_{1, t} \|\mtx{W}_{t-1}\mathbb{1}_{t-1}\|_{p,p}^2 + K_{2, t}\,,
\end{equation}
where
\begin{align*}
K_{1, t} & = (1+ 5 \ep_t^2) \left\{\left(\frac{1 + \eta_t \lambda_k}{1 + \eta_t \lambda_{k+1}}\right)^2 + 8 p \ep_t^2 \right\} \\
K_{2, t} & = 5 p k^{2/p} \ep_t^2\,.
\end{align*}

\end{proposition}
\begin{proof}
Reusing the notation of Proposition~\ref{j_bounds}, we have
\begin{equation*}
\mtx{W}_t \bone_{t-1} (\Id - {\mtx{\Delta}}_t^2) = \mtx{H}_t \bone_{t-1} + \mtx{J}_{t,1} \bone_{t-1} + \mtx{J}_{t,2} \bone_{t-1}\,,
\end{equation*}
where $\E[\mtx{J}_{t,1} \bone_{t-1} \mid \filtration_{t-1}] = \zeromtx$.
Since $\mtx{H}_t \bone_{t-1}$ is $\filtration_{t-1}$-measurable,
Proposition~\ref{non-centered-smoothness} therefore yields for any $\lambda > 0$
\begin{equation*}
\|\mtx W_t \bone_{t-1} (\Id - \mtx \Delta_t^2)\|_{p, p}^2 \leq (1 + \lambda)(\|\mtx H_t \bone_{t-1}\|_{p, p}^2 + (p-1) E_t^2 \ep_t^2 + \lambda^{-1} E_t^2 \ep_t^4)\,.
\end{equation*}
Choosing $\lambda = \ep_t^2$, we obtain
\begin{equation*}
\|\mtx W_t \bone_{t-1} (\Id - \mtx \Delta_t^2)\|_{p, p}^2 \leq (1 + \ep_t^2 )(\|\mtx H_t \bone_{t-1}\|_{p, p}^2 + p E_t^2 \ep_t^2)\,.
\end{equation*}
Finally, under the assumption that $\|\mtx \Delta_t \bone_{t-1}\| \leq \ep_t \leq \frac 12$ almost surely, on the event $\mathcal G_{t-1}$ the matrix $\Id - \mtx \Delta_t^2$ is invertible and satisfies
\begin{equation*}
\|(\Id - \mtx \Delta_t^2)^{-1} \bone_{t-1}\| \leq (1 - \|\mtx \Delta_t \bone_{t-1}\|^2)^{-1} \leq (1 - \ep_t^2)^{-1}
\end{equation*}
Hence
\begin{equation*}
\|\mtx W_t \bone_{t-1}\|_{p, p}^2 \leq \|\mtx W_t \bone_{t-1} (\Id - \mtx \Delta_t^2)\|_{p, p}^2 \|(\Id - \mtx \Delta_t)^{-1} \bone_{t-1}\| \leq \frac{1 + \ep_t^2}{(1 - \ep_t^2)^2}(\|\mtx H_t \bone_{t-1}\|_{p, p}^2 + p E_t^2 \ep_t^2)\,.
\end{equation*}
Since $\frac{1 + \ep_t^2}{(1 - \ep_t^2)^2} \leq 1 + 5 \ep_t^2$ for all $\ep_t \leq \frac 12$ and
\begin{equation*}
(1+5 \ep_t^2) E_t^2 \leq (1+5 \ep_t^2)(2k^{2/p} + 8 \|\mtx W_{t-1} \bone_{t-1}\|_{p, p}^2)
\end{equation*}
and $2(1+5 \ep_t^2) \leq 5$ for all $\ep_t \leq \frac 12$, this proves the claim.
\end{proof}

\section{Proof of Theorem~\ref{thm:main_recurrence}}\label{main-recurrence-proof}
We will unroll the one-step recurrence of Proposition~\ref{prop:one_step_recurrence}.
We first bound $K_{1, i}$.
We have
\begin{equation*}
K_{1, i} \leq \left(\frac{1 + \eta_i \lambda_k}{1 + \eta_i \lambda_{k+1}}\right)^2 + (5+8p) \ep_i^2 + 40p \ep_i^4 \leq \left(\frac{1 + \eta_i \lambda_k}{1 + \eta_i \lambda_{k+1}}\right)^2 + (5+18p) \ep_i^2\,,
\end{equation*}
where the second inequality follows from the first assumption in~\eqref{main_assumptions}.
The second assumption in~\eqref{main_assumptions} implies that $0 \leq 1 + \eta_i \lambda_k \leq 2$, so
\[\left(\frac{1 + \eta_i\lambda_{k+1}}{1+\eta_i\lambda_k}\right)^2 = \left(1 - \frac{\eta_i\rho_k}{1+\eta_i\lambda_k}\right)^2\leq \left(1 - \frac 12 \eta_i \rho_k\right)^2 \leq \econst^{-\eta_i \rho_k}\,.\]
Since $5 + 18p \leq 21p$ for all $p \geq 2$, we obtain
\begin{equation*}
K_{1, i} \leq \econst^{- \eta_i \rho_k} + C_1 p \ep_i^2\,.
\end{equation*}

We now proceed to prove the first claim by induction.
When $t = 1$, we use \eqref{eq:first_recursive} to obtain
\begin{align*}
\|\mtx W_1 \mathbb{1}_1\|_{p, p}^2 \leq \|\mtx W_1 \mathbb{1}_0\|_{p, p}^2 & \leq K_{1,1} \|\mtx W_0 \mathbb{1}_0\|_{p,p}^2 + K_{2, 1} \\
& \leq \econst^{- \eta_1 \rho_k} \|\mtx W_0 \bone_0\|_{p, p}^2 + C_1 p \ep_1^2 \|\mtx W_0 \bone_0\|_{p, p}^2 + C_2 pk^{2/p} \ep_1^2\,,
\end{align*}
which is the desired bound.

Proceeding by induction, for $t > 1$ we have
\begin{align*}
\|\mtx{W}_t\mathbb{1}_{t}\|_{p,p}^2 & \leq \|\mtx{W}_t\mathbb{1}_{t-1}\|_{p,p}^2 \\
& \leq K_{1,t} \|\mtx{W}_{t-1}\mathbb{1}_{t-1}\|_{p,p}^2 + K_{2,t}  \\
& \leq \econst^{-\eta_t \rho_k} \|\mtx W_{t-1} \bone_{t-1}\|_{p, p}^2 + C_1 p \ep_t^2 \|\mtx W_{t-1} \bone_{t-1}\|_{p, p}^2 + K_{2, t} \\
& \leq \econst^{-\eta_t \rho_k}\left(\econst^{-s_{t-1} \rho_k} \|\mtx W_{0} \bone_0\|_{p, p}^2 + C_1 p \ep_{t-1}^2 \sum_{i=0}^{t-2} \|\mtx W_{i} \bone_i\|_{p, p}^2 + C_2 p k^{2/p} \ep_{t-1}^2 (t-1) \right) \\
& \phantom{\leq} + C_1 p \ep_t^2 \|\mtx W_{t-1} \bone_{t-1}\|_{p, p}^2 + C_2 p k^{2/p} \ep_t^2 \\
& \leq \econst^{-s_t \rho_k} \|\mtx W_0 \bone_0\|_{p, p}^2 + C_1 p \ep_t^2 \sum_{i=0}^{t-1} \|\mtx W_{i} \bone_i\|_{p, p}^2 + C_2 p k^{2/p} \ep_t^2 t\,,
\end{align*}
where in the final inequality we have used that $\econst^{- \eta_t \rho_k} \ep_{t-1}^2 \leq \ep_t^2$ by the third assumption of~\eqref{main_assumptions}.
This proves the first bound.

For the second bound, we proceed in a similar way, but with a sharper bound on $K_{1, i}$.
The second assumption of~\eqref{main_assumptions} again implies
\[\left(\frac{1 + \eta_i\lambda_{k+1}}{1+\eta_i\lambda_k}\right)^2 = \left(1 - \frac{\eta_i\rho_k}{1+\eta_i\lambda_k}\right)^2\leq 1-\eta_i\rho_k + \frac{1}{4}(\eta_i\rho_k)^2\leq 1 - \frac{3}{4}\eta_i\rho_k\,,\]
and therefore
\begin{align*}
K_{1,i} &\leq (1 + 5 \ep_i^2) \left(1 - \frac{3}{4}\eta_i\rho_k + 8 p \ep_i^2 \right)\\
&\leq \exp\left(- \frac{3}{4}\eta_i\rho_k + (5 + 8p) \ep_i^2 \right)\\
&\leq \econst^{-\eta_i\rho_k/2},
\end{align*}
where the final step uses Assumption~\eqref{assume:p_small} and the fact that $5 + 8p \leq \frac{25}{2}p$ for all $p \geq 2$.

When $t = 1$, we therefore have
\begin{align*}
\|\mtx W_1 \mathbb{1}_1\|_{p, p}^2 \leq \|\mtx W_1 \mathbb{1}_0\|_{p, p}^2 & \leq K_{1,1} \|\mtx W_0 \mathbb{1}_0\|_{p,p}^2 + K_{2, 1} \\
& \leq \econst^{- \eta_1 \rho_k/2} \|\mtx W_0 \bone_0\|_{p, p}^2 + C_2 pk^{2/p} \ep_1^2\,,
\end{align*}
as desired, and for $t > 1$ the induction hypothesis yields
\begin{align*}
\|\mtx{W}_t\mathbb{1}_{t}\|_{p,p}^2 & \leq \|\mtx{W}_t\mathbb{1}_{t-1}\|_{p,p}^2 \\
& \leq K_{1,t} \|\mtx{W}_{t-1}\mathbb{1}_{t-1}\|_{p,p}^2 + K_{2,t}\\
& \leq \econst^{-\eta_t \rho_k/2}\left(\econst^{-s_{t-1} \rho_k/2} \|\mtx W_{0} \bone_0\|_{p, p}^2 +  C_2 p k^{2/p} \ep_{t-1}^2 (t-1) \right) \\
& \leq \econst^{-s_t \rho_k/2} \|\mtx W_0 \bone_0\|_{p, p}^2 +  C_2 p k^{2/p} \ep_t^2 t\,,
\end{align*}
where the final inequality again uses the third assumption in~\eqref{main_assumptions}.
This proves the second bound.
\qed

\section{Additional results for Section~\ref{sec:phase2}}
\begin{lemma}\label{phase2-assumption-verification}
Under the conditions of Proposition~\ref{prop:phase2_norm_bound}, the assumptions of~\eqref{main_assumptions} hold.
\end{lemma}

\begin{proof}
\noindent First assumption.  We have
\begin{equation*}
\ep_i = 2 \eta_i M(1+ \gamma) = 2(1 + \sqrt 2 \econst) \frac{\alpha M}{(\beta + i) \rho_k} \leq C_\ep \frac{\alpha}{\beta \barrho_k}\,,
\end{equation*}
where $C_\ep = 2 (1 + \sqrt 2 \econst)$.
So the first assumption is fulfilled as long as
\begin{subequations}
\label{alphabeta}
\begin{equation}\label{eq:alphabeta_1}
\beta/\alpha \geq 2 C_\ep/\barrho_k\,.
\end{equation}

\bigskip 

\noindent Second assumption.   
As above, we have
\begin{equation*}
\eta_i \|\mtx M\| \leq \frac{\alpha \|\mtx M\|}{\beta \rho_k} \leq \frac{\alpha}{\beta \barrho_k}\,,
\end{equation*}
so the assumption is fulfilled if~\eqref{eq:alphabeta_1} holds. 

\bigskip

\noindent Third assumption. 
It suffices to show that
\begin{equation*}
\frac{\ep_{i-1}}{\ep_i} \leq 1 + \frac{\eta_i \rho_k}{4} \quad \quad \forall i \geq 2\,,
\end{equation*}
which is equivalent to
\begin{equation*}
\frac{1}{\beta + i - 1} \leq \frac{\alpha/4}{\beta + i} \quad \quad \forall i \geq 2\,.
\end{equation*}
This holds as long as 
\begin{equation}\label{eq:alphabeta_3}
\alpha \geq 8\,.
\end{equation}
\end{subequations}
We obtain that all three assumptions hold under~\eqref{eq:alphabeta_1} and~\eqref{eq:alphabeta_3}, as claimed.

\end{proof}

\begin{lemma}\label{lem:phase1-last-bound}
In the setting of Theorem~\ref{thm:phase1}, if $s = 2 \econst \sqrt{\frac{\beta + 1}{\beta + T}}$, then
\begin{equation*}
\Prob{\|\mtx W_T\| \geq s} \leq \delta/2\,.
\end{equation*}
\end{lemma}
\begin{proof}
We have
\[\Prob{\|\mtx{W}_T \mathbb{1}_T\| \geq s} \leq \inf_{p \geq 2} s^{-p} \|\mtx{W}_T \mathbb{1}_T\|_{p,p}^p \,.\]
In particular, we choose 
\[s^2 = \econst^2\left(\frac{\beta+1}{\beta+T}\right)^\alpha+ \econst^2\frac{C_3^2 \alpha^2}{\barrho_k^2} \frac{T}{(\beta+T)^2}\log(2k/\delta),\quad \text{and}\quad p = \log(2k/\delta)\,.\]
It then follows from \eqref{eq:pq-normbound_decay} that
\[
\Prob{\|\mtx{W}_T \mathbb{1}_T\| \geq s} \leq s^{-p} \|\mtx{W}_T \mathbb{1}_T\|_{p}^p \leq k\left(\frac{1}{s^2}\left(\frac{\beta+1}{\beta+T}\right)^{\alpha} + \frac{1}{s^2}p  \frac{C_3^2 \alpha^2}{\barrho_k^2}\frac{T}{(\beta+T)^2} \right)^{p/2} =  k \econst^{-p} = \delta/2.\]
Combining the above bounds, we obtain that
\begin{equation*}
\|\mtx W_T\| \leq s \leq \econst \left(\frac{\beta+1}{\beta+T}\right)^{\alpha/2}+ \econst\frac{C_3 \alpha M}{\rho_k} \sqrt{\frac{\log(2k/\delta)}{T}}\,,
\end{equation*}
with probability at least $1 -\delta$.
Since both terms are smaller than $\econst\sqrt{\frac{\beta+1}{\beta + T}}$, the claim follows.
\end{proof}

\section{Additional results for Section~\ref{sec:phase1}}\label{sec:additional-phase1}
Our main tool will be the following slight variation on Proposition~\ref{j_bounds}.
\begin{proposition}\label{j_bounds_phasei}
Let $t \geq 1$.
Assume that $\eta_t$ is small enough that $\mtx M \succeq - \frac{1}{2 \eta_t} \Id$, and assume that~\eqref{eq:conditional_requirement} holds for $i = t$.
Consider an arbitrary deterministic matrix $\mtx E \in \mathcal E_{r, \ell}$.

Let
\begin{align*}
\bar E_t & = 1 + 2 \max_{\mtx E'' \in \mathcal E_{r+1, \ell + 1}} \|\mtx E'' \mtx W_{t-1} \bone_{t-1}\|_{p, p} \\
\ep & = 2 \eta M(1 + \gamma)\,.
\end{align*}

Then $\|\mtx \Delta_t \bone_{t-1}\| \leq \ep$ almost surely, and
\begin{equation}\label{eq:ew_decomp}
\mtx E \mtx{W}_t(\Id - {\mtx{\Delta}}_t^2) = \mtx E \mtx{H}_t + \mtx E \mtx{J}_{t,1} + \mtx E \mtx{J}_{t,2}
\end{equation}
for $\mtx E \mtx J_{t, 1}$ and $\mtx E\mtx J_{t, 2}$ satisfying
\begin{align*}
\|\mtx E \mtx{J}_{t,1} \bone_{t-1}\|_{p, p} & \leq \bar  E_t \ep  \\ 
\|\mtx E \mtx{J}_{t,2} \bone_{t-1}\|_{p, p} & \leq \bar E_t \ep^2\,,  \\ 
\end{align*}
and $\E[\mtx E \mtx J_{t, 1} \mid \filtration_{t-1}] = \zeromtx$.

\end{proposition}
\begin{proof}
The proof is a slight modification on the proof of Proposition~\ref{j_bounds}.
By construction, 
\begin{equation*}
\|\mtx E \mtx H_t\mathbb{1}_{t-1}\|_{p, p}^2 \leq \left(\frac{1 + \eta \lambda_k}{1 + \eta \lambda_{k+1}}\right)^2 \|\mtx E' \mtx W_{t-1}\mathbb{1}_{t-1}\|_{p, p}^2\,,
\end{equation*}
where $\mtx E' = \frac{1}{1 + \eta \lambda_{k+1}} \mtx E \mtx U^* (\Id + \eta \mtx \Sigma) \mtx U \in \mathcal E_{r+1, \ell} \subseteq \mathcal E_{r+1, \ell+1}$.

Similarly, we have
\begin{align*}
\|\mtx{E}\widehat{\mtx{\Delta}}_t \mathbb{1}_{t-1}\|_{p,p} & \leq 2\eta\|\mtx{E}\mtx{U}^*(\mtx{A}_t-\mtx{M}) \mtx{U}\mtx{W}_{t-1}\mathbb{1}_{t-1}\|_{p,p} + 2\eta\|\mtx{E}\mtx{U}^*(\mtx{A}_t-\mtx{M}) \mtx{V} \|_{p,p} \\
& \leq 2\eta  M( \|\mtx{E}''\mtx{W}_{t-1}\mathbb{1}_{t-1}\|_{p,p}+ \|\mtx{E}\|_{p,p})\\
& \leq \eps( \|\mtx{E}''\mtx{W}_{t-1}\mathbb{1}_{t-1}\|_{p,p}+ 1)\,
\end{align*}
where $\mtx E'' = \frac 1 M \mtx E \mtx U^* (\mtx A_t - \mtx M) \mtx U \in \mathcal E_{r + 1, \ell + 1}$, and we have used $\|\mtx{E}\|_{p}\leq \|\mtx E\|_2 \leq 1$.

We therefore obtain
\begin{align*}
\|\mtx E \mtx{J}_{t,1}\mathbb{1}_{t-1}\|_{p,p} & \leq \|\mtx E \widehat{\mtx{\Delta}}_t\mathbb{1}_{t-1}\|_{p,p} + \|\mtx E \mtx{H}_t\mathbb{1}_{t-1}\|_{p,p}\|{\mtx{\Delta}}_t\mathbb{1}_{t-1}\| \\
& \leq \left(\|\mtx E'' \mtx W_{t-1} \mathbb{1}_{t-1}\|_{p, p} + \|\mtx E' \mtx{W}_{t-1}\mathbb{1}_{t-1}\|_{p,p}  + 1 \right) \eps \\
& \leq \bar E_t \eps\,,
\end{align*}
and
\begin{equation*}
\|\mtx E \mtx{J}_{t,2}\mathbb{1}_{t-1}\|_{p,p}\leq\ \|\mtx E \widehat{\mtx{\Delta}}_t\mathbb{1}_{t-1}\|_{p, p}\|{\mtx{\Delta}}_t\mathbb{1}_{t-1}\| \leq  (\|\mtx E'' \mtx W_{t-1} \mathbb{1}_{t-1}\|_{p, p} + 1)\eps^2 \leq \bar E_t \eps^2\,.
\end{equation*}
\end{proof}
The following two results are the appropriate analogues of Proposition~\ref{prop:one_step_recurrence} and Theorem~\ref{thm:main_recurrence}.
\begin{proposition}\label{prop:one_step_prefix}
Adopt the setting of Proposition~\ref{j_bounds_phasei}.
If $\ep \leq 1/2$, then
\begin{equation}\label{eq:second_recursive}
\max_{\mtx E \in \mathcal E_{r, \ell}}\|\mtx E \mtx W_t \bone_{t-1}\|_{p, p}^2 \leq \bar  K_{1} \max_{\mtx E' \in \mathcal E_{r + 1, \ell}} \| \mtx E' \mtx W_{t-1}\bone_{t-1} \|_{p, p}^2 +  \bar K_{2} \max_{\mtx E'' \in \mathcal E_{r+1, \ell+1}} \|\mtx E'' \mtx W_{t-1} \bone_{t-1}\|_{p, p}^2 +  \bar K_{2}\,,
\end{equation}
where
\begin{align*}
\bar K_{1} & = (1 + 5 \ep^2) \left(\frac{1 + \eta \lambda_k}{1 + \eta \lambda_{k+1}}\right)^2 \\
\bar K_{2} & = (1 + 5 \ep^2) 8 p \ep^2
\end{align*}
\end{proposition}
\begin{proof}
As in the proof of Proposition~\ref{prop:one_step_recurrence}, we have for any $\mtx E \in \mathcal E_{r, \ell}$,
\begin{equation*}
\|\mtx E\|_{p, p}^2 \leq (1 + 5 \ep^2)(\|\mtx E \mtx H_t \bone_{t-1}\|_{p, p}^2 + p \bar E_t \ep^2)\,.
\end{equation*}
As in the proof of Proposition~\ref{j_bounds_phasei}, we can write
\begin{equation*}
\|\mtx E \mtx H_t\mathbb{1}_{t-1}\|_{p, p}^2 \leq \left(\frac{1 + \eta \lambda_k}{1 + \eta \lambda_{k+1}}\right)^2 \|\mtx E' \mtx W_{t-1}\mathbb{1}_{t-1}\|_{p, p}^2
\end{equation*}
where $\mtx E' = \frac{1}{1 + \eta \lambda_{k+1}} \mtx E \mtx U^* (\Id + \eta \mtx \Sigma) \mtx U \in \mathcal E_{r+1, \ell}$.
Since
\begin{equation*}
\bar E_t^2 \leq 8 \max_{\mtx E'' \in \mathcal E_{r+1, \ell+1}} \|\mtx E'' \mtx W_{t-1} \bone_{t-1}\|_{p, p}^2 + 8\,,
\end{equation*}
taking the maximum over all $\mtx E \in \mathcal E_{r, \ell}$ and $\mtx E' \in \mathcal E_{r+1, \ell}$ yields the claim.
\end{proof}

\begin{theorem}\label{thm:prefix_recurrence}
Let $t \leq T_0$ be a positive integer, and assume the following requirements hold for some $p \geq 2$:
\begin{subequations}
\begin{align}
\ep & \leq \frac 12\,,\label{assume:ep_small1}\\
\eta \|\mtx M\| & \leq \frac 12\,, \label{assume:eta_small1}\\
p \ep^2 & \leq \frac{\eta \rho_k}{50} \label{assume:p_small1} \\
\gamma & \geq 2 \,. \label{assume:gamma_big}
\end{align}
Then for any $r, \ell \in [T_0 - t + 1]$ and $p \geq 2$,
\begin{equation*}
\max_{\mtx E \in \mathcal E_{r, \ell}} \|\mtx E \mtx W_t \bone_t\|_{p, p}^2 \leq 
\max_{\mtx E \in \mathcal E_{r, \ell}} \|\mtx E \mtx W_t \bone_{t-1}\|_{p, p}^2 \leq \frac{\ell \gamma^2}{2 \econst^2} \econst^{ - t \eta  \rho_k/2} + C_4 p \gamma^2 \ep^2 t\,.
\end{equation*}
where $C_4 = 6$.
\end{subequations}

\end{theorem}
\begin{proof}
First, as in the proof of Theorem~\ref{thm:main_recurrence}, Assumptions~\eqref{assume:eta_small1} and~\eqref{assume:p_small1} imply
\begin{align*}
\bar K_1 + \bar K_2 & = (1 + 5 \ep^2)\left\{\left(\frac{1 + \eta \lambda_k}{1 + \eta \lambda_{k+1}}\right)^2 + 8 p \ep^2\right\} \\
& \leq \econst^{- \eta \rho_k/2}\,.
\end{align*}
In particular, $\bar K_1 + \bar K_2 \leq 1$.
Assumption~\eqref{assume:ep_small1} likewise implies that $\bar K_2 \leq 18$.

We now turn to the proof of the main claim, which we prove by induction on $t$.
For convenience, we introduce the notation $\gamma_\econst = \gamma/\sqrt{2}\econst$.
When $t = 1$ and $r,\ell \leq T_0$, \eqref{eq:second_recursive} implies 
\begin{align*}
\max_{\mtx E \in \mathcal E_{r, \ell}} \|\mtx E \mtx W_1 \mathbb{1}_1\|_{p, p}^2 &\leq \max_{\mtx E \in \mathcal E_{r, \ell}} \|\mtx E \mtx W_1 \mathbb{1}_{0} \|_{p, p}^2\\
& \leq  \bar K_1 \max_{\mtx E' \in \mathcal E_{r+1, \ell}} \|\mtx E' \mtx W_0 \mathbb{1}_0\|_{p, p}^2 +  \bar K_2 \max_{\mtx E'' \in \mathcal E_{r+1, \ell+1}} \|\mtx E'' \mtx W_0 \mathbb{1}_0\|_{p, p}^2 \eps^2 + \bar K_2 \\
& \leq  \bar K_1 \ell \gamma_\econst^2 +  \bar K_2 (\ell + 1) \gamma_\econst^2  +  \bar K_2 \\
& \leq \ell \gamma_\econst^2 ( \bar K_1 +  \bar K_2 )  + (1 + \gamma_\econst^2) \bar K_2 \\
& \leq \ell \gamma_\econst^2 \econst^{-\eta \rho_k/2} + \frac{\gamma^2}{3} \bar K_2
\end{align*}
where we have used the definition of $\mathcal{G}_0$ and where the last step uses~\eqref{assume:gamma_big}.
Proceeding by induction, we have
\begin{align*}
\max_{\mtx E \in \mathcal E_{r, \ell}} \|\mtx E \mtx W_t \mathbb{1}_t\|_{p, p}^2 & \leq \max_{\mtx E \in \mathcal E_{r, \ell}} \|\mtx E \mtx W_t \mathbb{1}_{t-1} \|_{p, p}^2 \\
& \leq  \bar K_{1} \max_{\mtx E' \in \mathcal E_{r+1, \ell}}  \|\mtx E' {\mtx W}_{t-1} \mathbb{1}_{t-1}\|_{p, p}^2 +  \bar K_{2} \max_{\mtx E'' \in \mathcal E_{r+1, \ell+1 }} \|\mtx E'' {\mtx W}_{t-1} \mathbb{1}_{t-1}\|_{p, p}^2  + \bar K_2  \\ 
& \leq  \bar K_1 (\ell  \gamma_\econst^2 \econst^{-(t-1) \eta \rho_k/2} + (t-1)\gamma^2 \bar K_2) \\
& \phantom{\leq} +  \bar K_2  ((\ell+1)  \gamma_\econst^2\econst^{-(t-1) \eta \rho_k/2} + (t-1)\gamma^2 \bar K_2) + \bar K_2 \\
& \leq \ell \gamma_\econst^2 (\bar K_1 + \bar K_2) \econst^{-(t-1)\eta \rho_k/2} + (t-1) (\bar K_1 + \bar K_2) \gamma^2 \bar K_2 + (1+ \gamma_\econst^2) \bar K_2 \\
& =  \ell \gamma_\econst^2 \econst^{-t \eta \rho_k/2} + \frac{\gamma^2}{3} \bar K_2 t\,,
\end{align*}
as claimed.
\end{proof}

\begin{proposition}\label{prop:phase1_moment_bounds}
Fix $s \in (0, 1)$, $2 \leq \gamma \leq C_\gamma \frac{d}{\delta^2}$, and $p \geq 2$, where $C_\gamma = 144 \gamma$ is the constant in Lemma~\ref{lem:gamma_0}.
Given $\rho > 0$, define the normalized gap
\begin{equation*}
\barrho = \min\left\{\frac{M}{\rho}, \frac{\|\mtx M\|}{\rho}, 1\right\}\,,
\end{equation*}
and adopt the step size
\begin{equation*}
\eta = \frac{C_\eta \log(\econst d/s\delta)}{\rho T_0}\,.
\end{equation*}
If $\rho_k \geq \rho/2$ and
\begin{equation*}
T_0 \geq p \cdot \frac{C_T \gamma^2 \log(\econst d/s\delta)^2}{s^2 \barrho^2}\end{equation*}
where
\begin{equation*}
C_\eta \geq 8 + 4 \log 2 C_\gamma \,, \quad \quad C_T \geq 600 \econst^2 C_\eta^2\,,
\end{equation*}
then
\begin{equation*}
\|\mtx W_{T_0} \bone_{T_0-1}\|_{p, p}^2 \leq \frac{s^2}{2 \econst^2}\left(1 + k^{2/p}\right)
\end{equation*}
and
\begin{equation*}
\max_{\mtx E \in \mathcal E_{1, 1}} \|\mtx E \mtx W_t \bone_{t-1}\|_{p, p} \leq \frac{\gamma}{\econst}
\end{equation*}
for all $1 \leq t \leq T_0$.

\end{proposition}
\begin{proof}
We will apply Theorems~\ref{thm:main_recurrence} and \ref{thm:prefix_recurrence}.
First, note that~\eqref{assume:gamma_big} holds by assumption.
We now turn to the other conditions.

\paragraph{Assumption~\eqref{assume:ep_small1}: }
Since $\gamma \geq 2$, we have
\begin{equation*}
\ep = 2 \eta M (1 + \gamma) \leq \frac{3 C_\eta \gamma \log(\econst d/s\delta) M}{\rho T_0}\,.
\end{equation*}
The assumption therefore holds as long as
\begin{equation}\label{eq:ct_assume1}
C_T \geq 3 C_\eta\,.
\end{equation}

\paragraph{Assumption~\eqref{assume:eta_small1}:}
As above, we have
\begin{equation*}
\eta \|\mtx M\| \leq \frac{2 C_\eta \log(\econst d/s\delta) \|\mtx M\|}{\rho T_0}\,,
\end{equation*}
and the requirement~\eqref{eq:ct_assume1} implies that this quantity is also smaller than $1/2$.
\paragraph{Assumption~\eqref{assume:p_small1}:}
Since $\eta \rho_k = \frac{C_\eta \log(\econst d/s\gamma)}{T_0} \geq \frac{1}{T_0}$ and $36 \econst^2 > 50$, it suffices to prove the stronger claim
\begin{equation}\label{assume:p_small_stronger}
p \ep^2 \leq \frac{s^2}{36 \econst^2 T_0}\,.
\end{equation}

This is satisfied so long as
\begin{equation*}
p \cdot \frac{16 C_\eta^2 \gamma^2 \log^2(\econst d/s\delta) M^2}{\rho^2 T_0^2} \leq \frac{s^2}{36 \econst^2 T_0}\,.
\end{equation*}
which will hold if
\begin{equation}\label{eq:ct_assume2}
C_T \geq 600 \econst^2 C_\eta^2\,.
\end{equation}
This requirement is stronger than~\eqref{eq:ct_assume1}, so Assumptions~\eqref{assume:ep_small1}--\eqref{assume:p_small1} hold under the sole condition~\eqref{eq:ct_assume2}.

We now turn to the two claimed bounds.
First, we instantiate Theorem~\ref{thm:main_recurrence} with the choice $\eta_i = \eta$ for $1 \leq i \leq T_0$.
The third assumption of~\eqref{main_assumptions} is trivially satisfied when when $\eta_i$ is constant, since in that case $\ep_i = \ep_{i-1}$ for all $i \geq 1$.
The remaining assumptions correspond directly to Assumptions~\eqref{assume:ep_small1},~\eqref{assume:eta_small1}, and~\eqref{assume:p_small1}.
The assumptions of Theorem~\ref{thm:main_recurrence} are therefore satisfied, so we obtain,
\begin{equation*}
\|\mtx W_{T_0} \bone_{T_0-1}\|_{p, p}^2 \leq \econst^{- T_0 \eta \rho_k/2} \|\mtx W_0 \bone_0\|_{p, p}^2 + 5 p k^{2/p} \ep^2 T_0\,. 
\end{equation*}
The definition of $\good_0$ in~\eqref{eq:phase1_goodb} and the fact that $\rho_k \geq \rho/2$ implies that the first term is at most
\begin{equation*}
\econst^{- T_0 \eta \rho_k/2} d \gamma^2 = (\econst d/s\delta)^{-C_\eta /4} d \gamma^2\,,
\end{equation*}
and this will be less than $\frac{s^2}{2 \econst^2}$ if
\begin{equation}\label{eq:ceta_assume}
C_\eta \geq 8 + 4 \log (2C_\gamma)\,.
\end{equation}
Since \eqref{assume:p_small_stronger} holds, the second term satisfies
\begin{equation*}
5 p k^{2/p} \ep^2 T_0 \leq \frac{5 s^2}{36 \econst^2} k^{2/p} < \frac{s^2}{2 \econst^2} k^{2/p}\,.
\end{equation*}
We obtain
\begin{equation*}
\|\mtx W_{T_0} \bone_{T_0-1}\|_{p, p}^2 \leq \frac{s^2}{2 \econst^2} \left(1 + k^{2/p}\right)\,,
\end{equation*}
as claimed.

For the second claim, we rely on Theorem~\ref{thm:prefix_recurrence}.
Assumptions~\eqref{assume:ep_small1}--\eqref{assume:gamma_big} having already been verified, we obtain for all $1 \leq t \leq T_0$,
\begin{equation*}
\max_{\mtx E \in \mathcal E_{1, 1}} \|\mtx E \mtx W_t \bone_{t-1}\|_{p, p}^2 \leq \frac{\gamma^2}{2 \econst^2} \econst^{- t \eta \rho_k/2} + 18 p \gamma^2 \ep^2 t\,.
\end{equation*}
Since $\rho_k \geq 0$, the first term is at most $\frac{\gamma^2}{2 \econst^2}$, and the second term is also at most $\frac{\gamma^2}{2 \econst^2}$ by~\eqref{assume:p_small_stronger}.
We obtain that
\begin{equation*}
\max_{\mtx E \in \mathcal E_{1, 1}} \|\mtx E \mtx W_t \bone_{t-1}\|_{p, p}^2 \leq \frac{\gamma^2}{\econst^2}\,,
\end{equation*}
as claimed.

\end{proof}

With Proposition~\ref{prop:phase1_moment_bounds} in hand, we can prove a full version of Theorem~\ref{thm:phase1_informal}.
\begin{theorem}\label{thm:phase1}
Fix a $\rho > 0$ and assume $|\mathrm{supp}(\lawa)| = m$.
Let
\begin{equation*}
\barrho = \max\left\{\frac{\rho}{M}, \frac{\rho}{\|\mtx M\|}, 1\right\}\,,
\end{equation*}
and set $s = 1/6$.

Adopt the step size
\begin{equation*}
\eta = \frac{C_\eta \log(\econst d/\delta s)}{\rho T_0}
\end{equation*}
where
\begin{equation*}
T_0 \geq \frac{C_T k (\log 12 \econst d/\delta \barrho s)^4}{s^2 \delta^2 \barrho^2}\,.
\end{equation*}

and
\begin{equation*}
C_\eta \geq 8 + 2 \log 144 C_\gamma \,, \quad \quad C_T \geq (12000 \econst^2 C_\eta^2C_ \gamma^2)^{5/4}\,.
\end{equation*}

If $m \leq T_0^3$ and $\rho_k \geq \rho/2$, then
\begin{equation*}
\|\mtx W_{T_0}\| \leq 1/6
\end{equation*}
with probability at least $1-\delta/3$.
\end{theorem}
\begin{proof}
We first show that we can assume that $\log T_0 \leq 5 \log(C_T d/\delta \barrho s)$.
Indeed, if $T_0 > \left(\frac{C_T d}{\delta \barrho s}\right)^5$, a crude argument similar to the one employed in the analysis of Phase II yields the claim.
We give the full details in Appendix~\ref{large_t0}.
In what follows, we therefore assume
\begin{equation}\label{assume:t_small}
\log T_0 \leq 5 \log(C_T d/\delta \barrho s)\,.
\end{equation}

Set
\begin{equation*}
\gamma = 144  C_\gamma \min \left\{\frac{\sqrt{21 k \log(C_T d/\delta\barrho s)}}{\delta}, \frac{d }{\delta^2}\right\}\,,
\end{equation*}
where $C_\gamma$ is as in Lemma~\ref{lem:gamma_0}.

Recall that our goal is to show $\|\mtx{W}_{T_0}\|\leq s$ with probability at least $1-\delta/3$. The failure probability can be bounded as 
\begin{equation*}
\Prob{\|\mtx{W}_{T_0}\| \geq s} \leq  \Prob{\|\mtx{W}_{T_0} \mathbb{1}_{T_0}\| \geq s} + \Prob{\mathcal G_{T_0}^C} \leq \inf_{p \geq 2} s^{-p} \|\mtx{W}_{T_0} \mathbb{1}_{T_0}\|_{p,p}^p + \Prob{\mathcal G_{T_0}^C}\,.
\end{equation*}

If we choose $p = \log(6k/\delta)$, then since $\log(C_T) \leq C_T^{1/5} \log(12)$ for any value of $C_T$, we have
\begin{align*}
T_0 & \geq \frac{C_T k(\log (12 \econst d/\delta \barrho s))^4}{s^2 \delta^2 \barrho^2} \\
& \geq  \log(6k/\delta) \cdot C_T^{4/5} \frac{k \log(C_T  d/\delta \barrho s)}{\delta^2} \cdot \frac{\log(\econst d/s \delta)^2}{s^2 \barrho^2} \\
& \geq p \frac{600 \econst^2 C_\eta^2  \gamma^2 \log(\econst d/s \delta)^2}{s^2 \barrho^2}\,,
\end{align*}
as long as
\begin{equation*}
C_T \geq (12000  \econst^2 C_\eta^2(144C_ \gamma)^2)^{5/4}\,,
\end{equation*}
which verifies the assumption of Proposition~\ref{prop:phase1_moment_bounds}.

We obtain
\[\|\mtx{W}_{T_0} \mathbb{1}_{T_0}\|_{p,p}^2 \leq \frac{s^2}{2 \econst^2}(1 + k^{2/p}) \leq k^{2/p} \frac{s^2}{\econst^2}\,.\]

We therefore have
\[ s^{-p} \|\mtx{W}_{T_0} \mathbb{1}_{T_0}\|_{p,p}^p\leq \econst^{-\log(6k/\delta)}\leq \delta/6\,.\]

It remains to bound $\Prob{\mathcal G_{T_0}^C}$.
Clearly
\begin{equation*}
\Prob{\mathcal G_{T_0}^C} \leq \Prob{\mathcal G_0^C} + \sum_{j = 1}^{T_0} \Prob{\mathcal G_j^C \cap \mathcal G_{j-1}}\,.
\end{equation*}
Since $m \leq T_0^3$ and we have assumed $\log T_0 \leq 5 \log(C_T d/\delta \barrho s)$, we have
\begin{equation*}
\log(\econst m T_0/\delta) \leq 4 \log(T_0) + \log(\econst/\delta) \leq 20 \log(C_T d/\delta \barrho s) + \log(\econst/\delta) \leq 21 \log(C_T d/\delta \barrho s)\,,
\end{equation*}
so Lemma~\ref{lem:gamma_0} guarantees that $\good_0$ holds with probability at least $1-\delta/12$.

For the second term, we have
\begin{equation*}
\Prob{\mathcal G_j^C \cap \mathcal G_{j-1}} = \Prob{\max_{\mtx{E}\in \mathcal{E}_{1,1}}\|\mtx{E} \mtx W_j \mathbb 1_{j-1}\| \geq \gamma}\leq \sum_{\mtx{E}\in \mathcal{E}_{1,1}}\Prob{\|\mtx{E} \mtx W_j \mathbb 1_{j-1}\| \geq \gamma}\,.
\end{equation*}

Choose $p = 21 \log(C_T d/\delta \barrho s)$.
The same argument as above yields
\begin{equation*}
T_0 \geq p \cdot C_T^{3/5} \frac{k \log(C_T d/\delta \barrho s)}{\delta^2} \cdot \frac{\log^2(\econst d/s \delta)}{s^2 \barrho^2}\,,
\end{equation*}
and this will be larger than the lower bound required on $T_0$ that was assumed in Proposition~\ref{prop:phase1_moment_bounds} as long as 
\begin{equation*}
C_T \geq (12000  \econst^3 C_\eta^2(144C_ \gamma)^2)^{5/3}
\end{equation*}

Proposition~\ref{prop:phase1_moment_bounds} therefore yields
\begin{equation*}
\Prob{\|\mtx{E} \mtx W_j \mathbb 1_{j-1}\| \geq \gamma} \leq \gamma^{-p}\|\mtx{E} \mtx W_j \mathbb 1_{j-1}\|_{p,p}^p \leq \econst^{-p} = \econst^{- 21 \log(C_T d/\delta \barrho s)} \quad \text{for all $\mtx{E}\in \mathcal{E}$}\,,
\end{equation*}
and thus
\[\Prob{\mathcal G_j^C | \mathcal G_{j-1}}\leq \sum_{\mtx{E}\in \mathcal{E}_{1,1}}\Prob{\|\mtx{E} \mtx W_j \mathbb 1_{j-1}\| \geq \gamma }\leq m\econst^{- 21 \log(C_T d/\delta \barrho s)} \,.\]
This yields
\begin{equation*}
\sum_{j=1}^{T_0} \Prob{\mathcal G_j^C | \mathcal{G}_{j-1}} \leq m T_0 \econst^{- 21 \log(C_T d/\delta \barrho s)} \leq \econst^{- 21 \log(C_T d/\delta \barrho s) + 4 \log T_0} \leq \delta/12\,,
\end{equation*}
where the last step uses~\eqref{assume:t_small}.
Finally, choosing $s = 1/6$, we obtain
\begin{equation*}
\Prob{\|\mtx{W}_{T_0}\| \geq 1/6} \leq \delta/3\,,
\end{equation*}
as claimed.
\end{proof}

\section{A reduction to finite support}\label{ghost}
Let $\Omega$ be the space of $d \times d$ symmetric matrices.
We argue that it suffices to assume that $P_A$ has finite support of cardinality at most $T_0^3$ in Phase I. We prove this by comparing the product measure $P_A^{\otimes T_0}$ with another distribution $P_m$ on $\Omega^{\otimes T_0}$. We specify this distribution by the following procedure: drawing a $T_0$-tuple $(A_1,\dots,A_{T_0})$ from the distribution $P_m$ is accomplished by
\begin{enumerate}
\item Drawing $m$ independent samples $\hat{\mtx{A}}_1,\dots,\hat{\mtx{A}}_m$ from $P_A$.
\item Drawing $\mtx{A}_1,\dots,\mtx{A}_{T_0}$ independently from the discrete distribution 
\[P_{\hat{A}} = \frac{1}{m}\sum_{i=1}^m\delta_{\hat{\mtx{A}}_i}.\]
That is, drawing $\mtx{A}_1,\dots,\mtx{A}_{T_0}$ independently and uniformly from the set $\{\hat{\mtx{A}}_i\}_{i=1}^m$ \textbf{with} replacement. 
\end{enumerate}
We will rely on the fact that the two distributions, $P_A^{\otimes T_0}$ and $P_m$, are close in total variation distance when $m$ is large. To see this, we first recognize that drawing $(A_1,\dots,A_{T_0})$ from $P_A^{\otimes T_0}$ is equivalent to the following: 
\begin{enumerate}
\item Draw $m$ independent samples $\hat{\mtx{A}}_1,\dots,\hat{\mtx{A}}_m$ from $P_A$.
\item Draw $\mtx{A}_1,\dots,\mtx{A}_{T_0}$ sequentially and uniformly from the set $\{\hat{\mtx{A}}_i\}_{i=1}^m$ \textbf{without} replacement. Denote by $P^{(T_0)}_{\hat{A}}$ the distribution of this sampling. 
\end{enumerate}
It is a standard result~\cite{Fre77} that, given any $\{\hat{A}_i\}_{i=1}^m$,  
\[d_{\text{TV}}\left(P_{\hat{A}}^{\otimes T_0},P^{(T_0)}_{\hat{A}}\right)\leq \frac{1}{2} \frac{T_0^2}{m}.\]
We thus have the following:

\begin{proposition}\label{prop:total_variation} For any $\delta\in(0,1)$, it holds that 
\[d_{\text{TV}}\left(P_m, P_A^{\otimes T_0}\right)\leq \delta\]
for all $m\geq T_0^2/2\delta$.
\end{proposition}
\begin{proof}
For any set $S\subset \Omega^{\otimes T_0}$, we have
\begin{align*} 
\left|P_m(S) - P_A^{\otimes T_0}(S)\right| &= \left|\E_{\hat{A}_i\sim P_A,1\leq i\leq m}\left[P_{\hat{A}}^{\otimes T_0}(S) - P^{(T_0)}_{\hat{A}}(S)\right]\right|\\
&\leq \E_{\hat{A}_i\sim P_A,1\leq i\leq m} \left|P_{\hat{A}}^{\otimes T_0}(S) - P^{(T_0)}_{\hat{A}}(S)\right|\\
&\leq \E_{\hat{A}_i\sim P_A,1\leq i\leq m} d_{\text{TV}}\left(P_{\hat{A}}^{\otimes T_0},P^{(T_0)}_{\hat{A}}\right)\\
&\leq \frac{1}{2} \frac{T_0^2}{m} \leq \delta.
\end{align*}
The claim follows from taking the maximum of $|P_m(S) - P_A^{\otimes T_0}(S)|$ over all subsets of $\Omega^{\otimes T_0}$.
\end{proof}

Given any $\hat{\mtx{A}}_1,\dots,\hat{\mtx{A}}_m$, define the empirical average 
\[\hat{\mtx{M}}_m:= \E_{A\sim P_{\hat{A}}} \mtx{A} =\frac{1}{m}\sum_{i=1}^m \hat{\mtx{A}}_i.\]
Denote by $\hat{\lambda}_1\geq \hat{\lambda}_2\geq \cdots \geq \hat{\lambda}_d$ the eigenvalues of $\hat{\mtx{M}}_m$, and write $\hat{\rho}_k = \hat{\lambda}_k - \hat{\lambda}_{k+1}$. Let $\hat{V}\in \mathbb{R}^{d\times k}$ be the orthogonal matrix whose columns are the leading $k$ eigenvectors of $\hat{\mtx{M}}_m$, and let $\hat{\mtx U}\in \mathbb{R}^{d\times (d-k)}$ be the orthogonal matrix consisting of the remaining eigenvectors. Standard results of matrix concentration implies that $\hat{\mtx M}_m$ is close to $\mtx M$. In particular, we have the following:

\begin{proposition}\label{prop:Bernstein_Wedin}
Suppose that $m\geq \frac{35M^2}{\rho_k^2}\log(2d/\delta)$. Let $\hat{\mtx{A}}_1,\dots,\hat{\mtx{A}}_m$ be drawn independently from $P_A$. Then it holds with probability at least $1-\delta$ that 
\[\|\hat{\mtx{M}}_m-\mtx{M}\| \leq \rho_k/4,\]
and, in particular, 
\[\hat{\rho}_k\geq \rho_k/2 \quad \text{and}\quad \|\mtx{U}^*\hat{\mtx{V}}\|\leq 1/3.\]
\end{proposition}
\begin{proof} By assumption 2, we have that $\|\hat{\mtx{M}}_m-\mtx{M}\|\leq M$ almost surely. Then the matrix Bernstein inequality~\cite[Theorem 1.4]{Tro12} implies that, for any $t\geq 0$, 
\[\Prob{\|\hat{\mtx{M}}_m-\mtx{M}\|\geq t} \leq 2d\exp\left(\frac{-mt^2/2}{M^2 + Mt/3}\right).\]
Substituting $t = \rho_k/4$ yields the first claim. Using the perturbation theory of eigenvalues of symmetric matrices, we have
\[\hat{\lambda}_k\geq \lambda_k - \|\hat{\mtx{M}}_m-\mtx{M}\| \quad \text{and} \quad \hat{\lambda}_{k+1}\leq \lambda_{k+1} - \|\hat{\mtx{M}}_m-\mtx{M}\|.\]
Therefore, conditioned on the first claim, it holds that
\[\hat{\rho}_k\geq \rho_k - 2\|\hat{\mtx{M}}_m-\mtx{M}\| \geq \frac{\rho_k}{2}.\]
Furthermore, it follows from Wedin's inequality~\cite{Wedin1972} that 
\[\|\mtx{U}^*\hat{\mtx{V}}\|\leq \frac{\|\hat{\mtx{M}}_m-\mtx{M}\|}{\hat{\lambda}_k - \lambda_{k+1}}\leq \frac{1}{3}.\]
\end{proof}

\begin{proposition}\label{prop:stability}
Let $\mtx{U}$ and $\mtx{V}$ be orthogonal matrices such that $\mtx{U}\mtx{U}^*+ \mtx{V}\mtx{V}^* = \Id$, and let $\hat{\mtx{U}}$ and $\hat{\mtx{V}}$ be matrices of the same size satisfying the same requirement. Suppose $\|\mtx{U}^*\hat{\mtx{V}}\|\leq 1/2$ and $\|\hat{\mtx{U}}^*\mtx{S}(\hat{\mtx{V}}^*\mtx{S})^{-1}\|\leq \gamma \leq 1$. Then 
\[\| \mtx{U}^* \mtx{S} (\mtx{V}^* \mtx{S})^{-1}\|\leq \frac{2+4\gamma}{3-2\gamma}.\]
\end{proposition}

\begin{proof} A direct calculation yields 
\begin{align*}
\|\mtx{U}^* \mtx{S} (\mtx{V}^* \mtx{S})^{-1}\| &= \|\mtx{U}^*( \hat{\mtx{U}}\hat{\mtx{U}}^*+\hat{\mtx{V}}\hat{\mtx{V}}^*) \mtx{S} (\mtx{V}^* \mtx{S})^{-1}\|\\
&\leq \|\hat{\mtx{U}}^*\mtx{S} (\mtx{V}^* \mtx{S})^{-1}\| + \|\mtx{U}^*\hat{\mtx{V}}\hat{\mtx{V}}^* \mtx{S} (\mtx{V}^* \mtx{S})^{-1}\|\\
&\leq \|\hat{\mtx{U}}^*\mtx{S}(\hat{\mtx{V}}^* \mtx{S})^{-1}\hat{\mtx{V}}^* \mtx{S}(\mtx{V}^* \mtx{S})^{-1}\| + \frac{1}{2}\|\hat{\mtx{V}}^* \mtx{S} (\mtx{V}^* \mtx{S})^{-1}\|\\
&\leq (\gamma + \frac{1}{2})\|\hat{\mtx{V}}^* \mtx{S} (\mtx{V}^* \mtx{S})^{-1}\|.
\end{align*}
We also have
\[\|\hat{\mtx{V}}^* \mtx{S} (\mtx{V}^* \mtx{S})^{-1}\| \leq \|\hat{\mtx{V}}^* \mtx{U}\mtx{U}^*\mtx{S} (\mtx{V}^* \mtx{S})^{-1}\| + \|\hat{\mtx{V}}^* \mtx{V}\mtx{V}^*\mtx{S} (\mtx{V}^* \mtx{S})^{-1}\| \leq \frac{1}{2}\|\mtx{U}^* \mtx{S} (\mtx{V}^* \mtx{S})^{-1}\| + 1.\]
Sequencing the two displays above and rearrange the inequality yields the claim.
\end{proof}

Now let $T_0$ be given as in Theorem~\ref{thm:phase1} and choose $m = T_0^2/2\delta$. As long as $T_0 \geq \frac{9 M}{\rho_k \delta} \log(d/\delta)$, we have
\[\frac{35M^2}{\rho_k^2}\log(2d/\delta) \leq m\leq T_0^3.\]
It then follows from Proposition \ref{prop:Bernstein_Wedin} that, when drawing $\hat{\mtx{A}}_1,\dots,\hat{\mtx{A}}_m$ independently from $P_A$, the event 
\begin{equation}\label{eq:good_event}
\mathcal{G} := \{\hat{\rho}_k\geq \rho_k/2 \ \text{and}\ \|\mtx{U}^*\hat{\mtx{V}}\|\leq 1/2\}
\end{equation}
happens with probability at least $1-\delta$. Conditioned on $\mathcal{G}$, we consider running $T_0$ steps of Oja's algorithm, with $A_1,\dots,A_{T_0}$ drawn i.i.d from $P_{\hat{A}}$. Note that the discrete distribution $P_{\hat{A}}$ also satisfies Assumption 1 and Assumption 2 (with $M$ replaced by $2M$). Our main theorem thus guarantees that with appropriately chosen step size, the output $Q_{T_0} = Q_{T_0}(A_1,\dots,A_{T_0})$ of this algorithm after $T_0$ steps satisfies
\[\|\hat{\mtx{U}}^*\mtx{Q}_{T_0}(\hat{\mtx{V}}^*\mtx{Q}_{T_0})^{-1}\|\leq \frac{1}{6}\]
with probability $1-\delta$. Combining \eqref{eq:good_event} and Proposition \ref{prop:stability}, we obtain that with probability at least $(1-\delta)^2\geq 1-2\delta$, the output of the algorithm satisfies
\[\|\mtx{U}^*\mtx{Q}_{T_0}(\mtx{V}^*\mtx{Q}_{T_0})^{-1}\|\leq 1,\]
that is,
\[P_m\left(\|\mtx{U}^*\mtx{Q}_{T_0}(\mtx{V}^*\mtx{Q}_{T_0})^{-1}\|\leq 1\right) \geq 1-2\delta.\]
Finally, we obtain from Proposition \ref{prop:total_variation} that
\begin{align*}
&P_m\left(\|\mtx{U}^*\mtx{Q}_{T_0}(\mtx{V}^*\mtx{Q}_{T_0})^{-1}\|\leq 1\right) \\
&\geq P_A^{\otimes T_0}\left(\|\mtx{U}^*\mtx{Q}_{T_0}(\mtx{V}^*\mtx{Q}_{T_0})^{-1}\|\leq 1\right) - d_{\text{TV}}\left(P_m, P_A^{\otimes T_0}\right)\\
&\geq 1-3\delta.
\end{align*}
In other words, with the same choice of $T_0$, the output of $T_0$ steps of Oja's algorithm with $A_1,\dots,A_{T_0}$ drawn i.i.d from the original distribution $P_A$ satisfies
\[\|\mtx{U}^*\mtx{Q}_{T_0}(\mtx{V}^*\mtx{Q}_{T_0})^{-1}\|\leq 1\]
with probability at least $1-3\delta$.

\section{Phase I succeeds if $T_0$ is large}\label{large_t0}
In this section, we prove Theorem~\ref{thm:phase1} when $T_0 > \frac{C_T^5 d^5}{\delta^5 \barrho^5 s^5}$.
Note that this value of $T_0$ is far larger than the optimal choice (which is of order $\tilde \Theta(k/\delta^2 \barrho^2 s^2)$), which makes the theorem much easier to prove.
Indeed, if $T_0$ is this large, we can prove Theorem~\ref{thm:phase1} directly by using the same conditioning argument as in Phase II.

\begin{proposition}
Assume $\eta$ and $T_0$ satisfy the requirements of Theorem~\ref{thm:phase1}, and assume $\rho \geq \rho_k/2$.
If~$T_0 \geq  \frac{C_T^5 d^5}{\delta^5 \barrho^5 s^5}$, then
\begin{equation*}
\|\mtx W_{T_0}\| \leq s
\end{equation*}
with probability at least $1-\delta/3$.
\end{proposition}
\begin{proof}
Set $\gamma = \frac{144 C_\gamma d}{\delta^2}$ where $C_\gamma$ is defined in Lemma~\ref{lem:gamma_0} and define the good events
\begin{align}\label{eq:phase1_good_tlarge}
\good_0 & := \{\|\mtx W_0\| \leq \gamma/(\sqrt{2} \econst)\}  \\
\good_i & := \{\|\mtx W_0\| \leq \gamma\} \cap \good_{i-1}\,, \quad \forall i \geq 1\,.
\end{align}
In order to apply Theorem~\ref{thm:main_recurrence}, we verify~\eqref{main_assumptions}

\paragraph{First assumption}
We have
\begin{equation*}
\ep = 2 \eta M (1 + \gamma) \leq \frac{3 C_\eta \log(\econst d/\delta s) M \gamma}{\rho T_0}\,,
\end{equation*}
and this quantity is smaller than $1/2$ so long as
\begin{equation}\label{eq:assume_ctpp}
C_T^5 \geq 864 C_\eta C_\gamma\,.
\end{equation}

\paragraph{Second assumption}
We again have
\begin{equation*}
\eta \|\mtx M\| = \frac{C_\eta \log(\econst d/\delta s) \|\mtx M\|}{\rho T_0}\,,
\end{equation*}
and~\eqref{eq:assume_ctpp} guarantees that this quantity is smaller than $1/2$ as well.

\paragraph{Third assumption}
Since $\ep_i = \ep$ for all $i$ and $\eta \rho \geq 0$, this requirement trivially holds.

Our goal is to bound
\begin{equation*}
\Prob{\|\mtx W_{T_0}\| \geq s} \leq \Prob{\|\mtx W_{T_0} \bone_{T_0}\| \geq s} + \Prob{\good_0^C} + \sum_{j=1}^{T_0} \Prob{\good_j^C \cap \good_{j-1}}\,.
\end{equation*}

Having verified~\eqref{main_assumptions}, we can employ~\eqref{eq:main_bound1}, obtaining
\begin{equation*}
\|\mtx W_{T_0} \bone_{T_0} \|_{p, p}^2 \leq \econst^{- T_0 \eta \rho_k} k^{2/p} \gamma^2/2 \econst^2 + (C_1 \gamma^2 + C_2) p k^{2/p} \ep^2 T_0\,.
\end{equation*}
For the first term, the fact that $\rho_k \geq \rho/2$ implies that 
\begin{equation*}
\econst^{- T_0 \eta \rho_k} \frac{\gamma^2}{2 \econst^2} = (\delta s/\econst d)^{C_\eta/2} \frac{\gamma^2}{2 \econst^2}\,,
\end{equation*}
and this is smaller than $\frac{s^2}{2 \econst^2}$ as long as
\begin{equation*}
C_\eta \geq 8 + 2 \log (144 C_\gamma)\,.
\end{equation*}

Letting $C_3$ be as in Proposition~\ref{prop:phase2_norm_bound} and choosing $p = \log(6 k/d\delta)$, we also have
\begin{equation*}
p (C_1 \gamma^2 + C_2) \ep^2 T_0\leq p \frac{144^2 C_3^2 C_\eta^2 \log^2 (\econst d/\delta s) M^2 \gamma^2}{\rho^2 T_0} \leq \frac{144^2 C_3^2 C_\eta^2 C_\gamma^2 \log^3(6 d/\delta s)}{C_T^5} \cdot \frac{\delta s}{d}
\end{equation*}
Since $\log^3(6 d/\delta s) \leq 9 \frac{d}{\delta s}$ for all positive $d$, $\delta$, and $s$, this quantity will be less than $\frac{s^2}{2 \econst^2}$ so long as
\begin{equation}\label{eq:assume_ctpp2}
C_T^5 \geq 2 (432 \econst C_3 C_\eta C_\gamma)^2\,,
\end{equation}
and this requirement subsumes~\eqref{eq:assume_ctpp}.

We therefore obtain, for $p = \log(6k/\delta)$,
\begin{equation*}
\Prob{\|\mtx W_0 \bone_{T_0}\| \geq s} \leq s^{-p}\|\mtx W_0 \bone_{T_0}\|_{p, p}^p \leq k \econst^{-p} \leq \delta/6\,,
\end{equation*}

In a similar way,~\eqref{eq:main_bound1} yields for all $t \in [T_0]$,
\begin{equation*}
\gamma^{-2} \|\mtx W_t \bone_{t-1}\|_{p, p}^2 \leq \frac{k^{2/p}}{2 \econst^2} + (C_1 \gamma^2 + C_2) p k^{2/p} \ep^2 T_0\,.
\end{equation*}
If we choose $p = \log(12 k T_0/\delta)$, then we have
\begin{equation*}
p (C_1 \gamma^2 + C_2) \ep^2 T_0\leq p \frac{C_3^2 C_\eta^2 \log^2 (\econst d/\delta s) M^2 \gamma^2}{\rho^2 T_0} \leq \frac{2 144^2 C_3^2 C_\eta^2 C_\gamma^2 \log^3(T_0)}{C_T^4 T_0^{1/5}}\,,
\end{equation*}
and since $\log^3(T_0) \leq 169 T_0^{1/5}$ for all $T_0$, we have that this quantity will be at most $\frac{1}{2 \econst^2}$ if
\begin{equation}\label{eq:assume_ctpp3}
C_T^5 \geq (3744 \econst C_3 C_\eta C_\gamma)^{5/2}\,,
\end{equation}
and this requirement subsumes~\eqref{eq:assume_ctpp2}, and it holds under the assumptions of Theorem~\ref{thm:phase1}.

By Lemma~\ref{lem:gamma_0}, the event $\good_0$ holds with probability at least $1-\delta/12$.

Finally, we have for any $j \in [T_0]$,
\begin{equation*}
\Prob{\good_j^C \cap \good_{j-1}} \leq \Prob{\|\mtx W_j \bone_{j-1}\| \geq \gamma} \leq \inf_{p \geq 2} \gamma^{-p} \|\mtx W_t \bone_{t-1}\|_{p, p}^p\,,
\end{equation*}
and choosing $p = \log(12 k T_0/\delta)$ we have
\begin{equation*}
\gamma^{-p} \|\mtx W_t \bone_{t-1}\|_{p, p}^p \leq k \econst^{-p} \leq \frac{12}{\delta T_0}\,,
\end{equation*}
and summing these probabilities for $j \in [T_0]$, yields that
\begin{equation*}
\Prob{\|\mtx W_{T_0}\| \geq s} \leq \Prob{\|\mtx W_{T_0} \bone_{T_0}\| \geq s} + \Prob{\good_0^C} + \sum_{j=1}^{T_0} \Prob{\good_j^C \cap \good_{j-1}} \leq \frac 1 6 + \frac 1{12} + \frac 1{12} = \frac 13\,,
\end{equation*}
as claimed.
\end{proof}
\section{Omitted proofs}\label{sec:omit}
\subsection{Proof of Lemma~\ref{lem:decomposition}}
We will show that
\begin{equation}
\mtx{W}_t(\Id - {\mtx{\Delta}}_t^2) = \mtx{H}_t + \mtx{J}_{t,1} + \mtx{J}_{t,2}\,,
\end{equation}
where 
\[\mtx{H}_t = \mtx{U}^*(\Id + \eta_t\mtx{M}){\mtx{Z}}_{t-1}(\mtx{V}^*(\Id + \eta_t\mtx{M}){\mtx{Z}}_{t-1})^{-1},\quad \mtx{J}_{t,1} = \widehat{\mtx{\Delta}}_t - \mtx{H}_t\mtx{\Delta}_t,\quad \text{and}\quad \mtx{J}_{t,2} = - \widehat{\mtx{\Delta}}_t\mtx{\Delta}_t\]
and where we write
\begin{equation*}
\widehat{\mtx{\Delta}}_t = \eta_t\mtx{U}^*(\mtx{A}_t-\mtx{M}){\mtx{Z}}_{t-1}(\mtx{V}^*(\Id + \eta_t\mtx{M}){\mtx{Z}}_{t-1})^{-1}\,.
\end{equation*}

By the definition of $\mtx Z_t$, we have
\[\mtx{W}_t = \mtx{U}^*{\mtx{Z}}_t(\mtx{V}^*{\mtx{Z}}_t)^{-1} = \mtx{U}^*\mtx{Y}_t\mtx{Z}_{t-1}(\mtx{V}^*\mtx{Y}_t\mtx{Z}_{t-1})^{-1}.\]

We have 
\begin{align*}
\mtx{V}^*\mtx{Y}_t{\mtx{Z}}_{t-1} =&\ \mtx{V}^*(\Id + \eta_t\mtx{M}){\mtx{Z}}_{t-1} + \eta_t\mtx{V}^*(\mtx{A}_t-\mtx{M}){\mtx{Z}}_{t-1}\\
=&\ \left(\Id + \eta_t\mtx{V}^*(\mtx{A}_t-\mtx{M}){\mtx{Z}}_{t-1}(\mtx{V}^*(\Id + \eta_t\mtx{M}){\mtx{Z}}_{t-1})^{-1}\right)\mtx{V}^*(\Id + \eta_t\mtx{M}){\mtx{Z}}_{t-1}\\
=&\ (\Id + {\mtx{\Delta}}_t)\mtx{V}^*(\Id + \eta_t\mtx{M}){\mtx{Z}}_{t-1},
\end{align*}
which implies
\begin{align*}
(\mtx{V}^*\mtx{Y}_t{\mtx{Z}}_{t-1})^{-1}(\Id - \mtx \Delta_t^2) =&\ (\mtx{V}^*(\Id + \eta_t\mtx{M}){\mtx{Z}}_{t-1})^{-1}(\Id + {\mtx{\Delta}}_t)^{-1}(\Id + {\mtx{\Delta}}_t)(\Id - {\mtx{\Delta}}_t)\\
=&\ (\mtx{V}^*(\Id + \eta_t\mtx{M}){\mtx{Z}}_{t-1})^{-1}(\Id - {\mtx{\Delta}}_t)\,.\end{align*}
We also have
\begin{align*}
\mtx{U}^*\mtx{Y}_t{\mtx{Z}}_{t-1} &= \mtx{U}^*(\Id + \eta_t\mtx{M}){\mtx{Z}}_{t-1} + \eta_t\mtx{U}^*(\mtx{A}_t-\mtx{M}){\mtx{Z}}_{t-1}\\
&= \mtx{U}^*(\Id + \eta_t\mtx{M}){\mtx{Z}}_{t-1} +\widehat{\mtx{\Delta}}_t(\mtx{V}^*(\Id + \eta_t\mtx{M}){\mtx{Z}}_{t-1}).
\end{align*} 
Therefore
\begin{align*}
\mtx{W}_t (\Id - \mtx \Delta_t^2) =&\ \mtx{U}^*\mtx{Y}_t{\mtx{Z}}_{t-1}(\mtx{V}^*\mtx{Y}_t{\mtx{Z}}_{t-1})^{-1} \\
=&\ \mtx{U}^*(\Id + \eta_t\mtx{M}){\mtx{Z}}_{t-1}(\mtx{V}^*(\Id + \eta_t\mtx{M}){\mtx{Z}}_{t-1})^{-1}\\
&\ + \widehat{\mtx{\Delta}}_t  - \mtx{U}^*(\Id + \eta_t\mtx{M}){\mtx{Z}}_{t-1}(\mtx{V}^*(\Id + \eta_t\mtx{M}){\mtx{Z}}_{t-1})^{-1}{\mtx{\Delta}}_t \\
&\ - \widehat{\mtx{\Delta}}_t{\mtx{\Delta}}_t\,.
\end{align*}
That is 
\begin{equation}
\mtx{W}_t(\Id - \widehat{\mtx{\Delta}}_t^2) = \mtx{H}_t + \mtx{J}_{t,1} + \mtx{J}_{t,2}\,.
\end{equation}
Since $\mtx \Delta_t$ and $\widehat{\mtx \Delta}_{t}$ are both $O(\eta_t)$, the claim follows.
\qed

\subsection{Proof of Proposition~\ref{non-centered-smoothness}}
By the triangle inequality, we have
\begin{equation*}
\|\mtx X + \mtx Y + \mtx Z\|_{p, p} \leq \|\mtx X + \mtx Y\|_{p, p} + \|\mtx Z\|_{p, p}\,,
\end{equation*}
which implies
\begin{align*}
\|\mtx X + \mtx Y + \mtx Z\|_{p, p}^2 & \leq (\|\mtx X + \mtx Y\|_{p, p} + \|\mtx Z\|_{p, p})^2 \\
& \leq (1+\lambda)(\|\mtx X + \mtx Y\|_{p, p}^2 + \lambda^{-1}\|\mtx Z\|^2_{p, p})\,,
\end{align*}
where in the second step we have applied the elementary inequality
\begin{equation*}
(a + b)^2 \leq (1 + \lambda)(a^2 + \lambda^{-1} b^2)\,,
\end{equation*}
valid for all real numbers $a$ and $b$ and $\lambda > 0$.
Applying Proposition~\ref{prop:smooth} to $\|\mtx X + \mtx Y\|_{p, p}^2$ then yields the claim.
\qed

\section{Additional Lemmas}
\begin{lemma}\label{lem:Gaussian_1}
For any deterministic matrices $\mtx A, \mtx B$ and any standard Gaussian matrix $\mtx Z$ of suitable sizes, it holds that
\[\Prob{\|\mtx{A}\mtx{Z}\mtx{B}\|_2\geq \|\mtx{A}\|_2\|\mtx{B}\|_2(1+t)}\leq \econst^{-t^2/2}.\]
\end{lemma}

\begin{proof}
Let $f(\mtx{X}) := \|\mtx{A}\mtx{X}\mtx{B}\|_2$, then
\[|f(\mtx{X}_1)-f(\mtx{X}_2)|\leq \|\mtx{A}\|\|\mtx{B}\|\cdot\|\mtx{X}_1-\mtx{X}_2\|_2.\]
By Gaussian concentration, we have
\[\Prob{f(\mtx Z)\geq \E f(\mtx Z) + \|\mtx{A}\|\|\mtx{B}\| t}\leq \econst^{-t^2/2}.\]
Moreover, we have
\[\E f(\mtx Z) \leq (\E \|\mtx{A}\mtx{Z}\mtx{B}\|_2^2)^{1/2} = \|\mtx{A}\|_2\|\mtx{B}\|_2.\]
It thus follows that 
\[\Prob{f(\mtx Z)\geq \|\mtx{A}\|_2\|\mtx{B}\|_2(1+t)}\leq \Prob{f(\mtx Z)\geq \E f(\mtx Z) + \|\mtx{A}\|\|\mtx{B}\| t}\leq \econst^{-t^2/2},\]
which is the stated result.
\end{proof}

\begin{lemma}[{\cite[Theorem II.13]{Davidson2001}}]\label{lem:Gaussian_2}
Let $\mtx{Q}\in \mathbb{R}^{d\times k}$ be a standard Gaussian matrix. Then
\[\Prob{\|\mtx{Q}\|\geq \sqrt{d} + \sqrt{k} + t}\leq 2\cdot \econst^{-t^2/2}\,.\]
\end{lemma}

\begin{lemma}[{\cite[Lemma i.A.3]{AllLi17}}]\label{lem:Gaussian_3}
Let $\mtx{Q}\in \mathbb{R}^{k\times k}$ be a standard Gaussian matrix. Then for every $\delta\in (0,1)$, 
\[\Prob{\|\mtx{Q}^{-1}\|_2\geq \frac{6 \sqrt{k}}{\delta}}\leq \delta.\]
\end{lemma}

The next lemma bounds the probability of $\mathcal{G}_0$ from below.

\begin{lemma}\label{lem:gamma_0}
Let $\good_0$ be the event defined in~\eqref{eq:phase1_goodb}.
There exists a positive constant $C_\gamma = 144 \econst$ such that for any $\delta \in (0, 1)$, if $\gamma \geq C_\gamma \min\{\sqrt{k\log(\econst mT_0/\delta)}/\delta, d/\delta^2\}$, then $\mathcal{G}_0$ holds with probability at least $1-\delta$.
\end{lemma}
\begin{proof}
We have $\mtx W_0 = \mtx U^* \mtx Z_0 (\mtx V^* \mtx Z_0)^{-1}$, where $\mtx Z_0 $ is a matrix with i.i.d.~Gaussian entries. Since $\mtx U$ and $\mtx V$ have orthonormal columns and are themselves orthogonal, the two matrices $\mtx V^* \mtx Z_0$ and $\mtx U^* \mtx Z_0$ are independent matrices with i.i.d.~Gaussian entries. Using Lemma~\ref{lem:Gaussian_1} and conditioning on $\mtx V^* \mtx Z_0$, we have that with probability at least $1-\delta/3(T_0+1)^2$, 
\begin{equation}\label{eq:event_Erl}
\max_{\mtx{E}\in \mathcal{E}_{r,\ell}}\|\mtx{E}\mtx U^* \mtx Z_0 (\mtx V^* \mtx Z_0)^{-1}\|_2 \leq \|(\mtx V^* \mtx Z_0)^{-1}\|_2 \cdot 2\sqrt{8 \ell \log(\econst mT_0/\delta)},
\end{equation}
where we have taken a union bound over the fewer than $((m+1) (T_0+1))^\ell$ elements of $\mathcal E_{r,\ell}$. Taking a uniform bound again over all $r,\ell\in[T_0 + 1]$ yields that, with probability at least $1-\delta/3$, the event \eqref{eq:event_Erl} holds for all $r,\ell\in [T_0 + 1]$. 
By Lemma~\ref{lem:Gaussian_3}, we also have that that $\|(\mtx V^* \mtx Z_0)^{-1}\|_2 \leq 18\sqrt{k}/\delta$ with probability at least $1-\delta/3$.
Furthermore, Lemma~\ref{lem:Gaussian_2} implies that $\|\mtx U^* \mtx Z_0\| \leq 2 \sqrt{2 d\log(3/\delta)}$ with probability at least $1-\delta/3$.
Combining these bounds, we obtain that with probability at least $1-\delta/$,
\begin{equation*}
\max_{\mtx{E}\in \mathcal{E}_{r,\ell}}\|\mtx{E}\mtx U^* \mtx Z_0 (\mtx V^* \mtx Z_0)^{-1}\|_2 \leq 36\sqrt{8 \ell \log(\econst mT_0/\delta)}\,,
\end{equation*}
which is less than $\frac{\sqrt{\ell} \gamma}{\sqrt{2} \econst}$ as long as $C_\gamma \geq 144 \econst$,
and under this same assumption
\begin{equation*}
\|\mtx W_0\|_2 \leq \|\mtx U^* \mtx Z_0\| \|(\mtx V^* \mtx Z_0)^{-1}\|_2 \leq 36\sqrt{2 d \log(3/\delta)} \leq \sqrt{d} \gamma
\end{equation*}
as well.

So $\good_0$ holds with probability at least $1-\delta$ if $\gamma \geq C_\gamma \sqrt{k\log(\econst mT_0/\delta)}/\delta$ for $C_\gamma \geq 144 \econst$.

On the other hand,
We have $\E \|\mtx U^* \mtx Z_0 \| \leq 2 \sqrt{d}$, so that $\|\mtx U^* \mtx Z_0\| \leq 4 \sqrt{d}/\delta$ with probability at least $1- \delta/2$, and Lemma~\ref{lem:Gaussian_3} implies that $\|\mtx V^* \mtx Z_0\|_2 \leq 12 \sqrt{k}/\delta$ with probability at least $1-\delta/2$, so with probability at least $1 - \delta$ we have
\begin{equation*}
\|\mtx W_0\|_2 \leq \|\mtx U^* \mtx Z_0\|\|(\mtx V^* \mtx Z_0)^{-1}\|_2 \leq 48 \sqrt{dk}/\delta^2 < 50 d/\delta^2\,.
\end{equation*}
as claimed.
On this event, we also have $\|\mtx E \mtx W_0\|_2 \leq \|\mtx W_0\|_2 \leq 50 d/\delta^2$.
Therefore, if $\gamma \geq 50\sqrt{2} \econst d/\delta^2$, then $\good_0$ holds.

So $\good_0$ holds with probability at least $1-\delta$ if $\gamma \geq C_\gamma d/\delta^2$ for $C_\gamma \geq 50 \sqrt{2} \econst$.
Therefore, taking $C_\gamma = 144 \econst$ satisfies both requirements and proves the claim.
\end{proof}

\bibliographystyle{habbrv}
\bibliography{oja.bib}

\begin{thebibliography}{10}

\bibitem{AllLi17}
Z.~Allen-Zhu and Y.~Li.
\newblock First efficient convergence for streaming {$k$}-{PCA}: a global,
  gap-free, and near-optimal rate.
\newblock In {\em 58th {A}nnual {IEEE} {S}ymposium on {F}oundations of
  {C}omputer {S}cience---{FOCS} 2017}, pages 487--492. IEEE Computer Soc., Los
  Alamitos, CA, 2017.

\bibitem{BalDuWan16}
M.~Balcan, S.~S. Du, Y.~Wang, and A.~W. Yu.
\newblock An improved gap-dependency analysis of the noisy power method.
\newblock In Feldman et~al. \cite{DBLP:conf/colt/2016}, pages 284--309.

\bibitem{DBLP:conf/icml/2016}
M.~Balcan and K.~Q. Weinberger, editors.
\newblock {\em Proceedings of the 33nd International Conference on Machine
  Learning, {ICML} 2016, New York City, NY, USA, June 19-24, 2016}, volume~48
  of {\em {JMLR} Workshop and Conference Proceedings}. JMLR.org, 2016.

\bibitem{BalDasFre13}
A.~Balsubramani, S.~Dasgupta, and Y.~Freund.
\newblock The fast convergence of incremental {PCA}.
\newblock In Burges et~al. \cite{DBLP:conf/nips/2013}, pages 3174--3182.

\bibitem{DBLP:conf/nips/2013}
C.~J.~C. Burges, L.~Bottou, Z.~Ghahramani, and K.~Q. Weinberger, editors.
\newblock {\em Advances in Neural Information Processing Systems 26: 27th
  Annual Conference on Neural Information Processing Systems 2013. Proceedings
  of a meeting held December 5-8, 2013, Lake Tahoe, Nevada, United States},
  2013.

\bibitem{Davidson2001}
K.~R. Davidson and S.~J. Szarek.
\newblock Local operator theory, random matrices and {B}anach spaces.
\newblock In {\em Handbook of the geometry of {B}anach spaces, {V}ol. {I}},
  pages 317--366. North-Holland, Amsterdam, 2001.

\bibitem{DavKah70}
C.~Davis and W.~M. Kahan.
\newblock The rotation of eigenvectors by a perturbation. {III}.
\newblock {\em SIAM J. Numer. Anal.}, 7:1--46, 1970.

\bibitem{DoaVav16}
X.~V. Doan and S.~Vavasis.
\newblock Finding the largest low-rank clusters with {K}y {F}an
  {$2$}-{$k$}-norm and {$\ell_1$}-norm.
\newblock {\em SIAM J. Optim.}, 26(1):274--312, 2016.

\bibitem{EdeAriSmi99}
A.~Edelman, T.~A. Arias, and S.~T. Smith.
\newblock The geometry of algorithms with orthogonality constraints.
\newblock {\em SIAM J. Matrix Anal. Appl.}, 20(2):303--353, 1999.

\bibitem{DBLP:conf/colt/2016}
V.~Feldman, A.~Rakhlin, and O.~Shamir, editors.
\newblock {\em Proceedings of the 29th Conference on Learning Theory, {COLT}
  2016, New York, USA, June 23-26, 2016}, volume~49 of {\em {JMLR} Workshop and
  Conference Proceedings}. JMLR.org, 2016.

\bibitem{Fre77}
D.~Freedman.
\newblock A remark on the difference between sampling with and without
  replacement.
\newblock {\em J. Amer. Statist. Assoc.}, 72(359):681, 1977.

\bibitem{GolVan96}
G.~H. Golub and C.~F. Van~Loan.
\newblock {\em Matrix computations}.
\newblock Johns Hopkins Studies in the Mathematical Sciences. Johns Hopkins
  University Press, Baltimore, MD, third edition, 1996.

\bibitem{HarPri14}
M.~Hardt and E.~Price.
\newblock The noisy power method: {A} meta algorithm with applications.
\newblock In Z.~Ghahramani, M.~Welling, C.~Cortes, N.~D. Lawrence, and K.~Q.
  Weinberger, editors, {\em Advances in Neural Information Processing Systems
  27: Annual Conference on Neural Information Processing Systems 2014, December
  8-13 2014, Montreal, Quebec, Canada}, pages 2861--2869, 2014.

\bibitem{HenWar19}
A.~Henriksen and R.~Ward.
\newblock Adaoja: Adaptive learning rates for streaming pca.
\newblock 05 2019, 1905.12115.

\bibitem{HenWar20}
A.~Henriksen and R.~Ward.
\newblock Concentration inequalities for random matrix products.
\newblock {\em Linear Algebra Appl.}, 594:81--94, 2020.

\bibitem{HuaNilTro20}
D.~Huang, J.~Niles-Weed, J.~A. Tropp, and R.~Ward.
\newblock Matrix concentration for products.
\newblock 03 2020, 2003.05437.

\bibitem{JaiJinKak16}
P.~Jain, C.~Jin, S.~M. Kakade, P.~Netrapalli, and A.~Sidford.
\newblock Streaming {PCA:} matching matrix bernstein and near-optimal finite
  sample guarantees for oja's algorithm.
\newblock In Feldman et~al. \cite{DBLP:conf/colt/2016}, pages 1147--1164.

\bibitem{Jolliffe2002}
I.~T. Jolliffe.
\newblock {\em Principal component analysis}.
\newblock Springer Series in Statistics. Springer-Verlag, New York, second
  edition, 2002.

\bibitem{JudNem08}
A.~Juditsky and A.~S. Nemirovski.
\newblock Large deviations of vector-valued martingales in 2-smooth normed
  spaces.
\newblock 09 2008, 0809.0813.

\bibitem{LiLinLu16}
C.~Li, H.~Lin, and C.~Lu.
\newblock Rivalry of two families of algorithms for memory-restricted streaming
  {PCA}.
\newblock In A.~Gretton and C.~C. Robert, editors, {\em Proceedings of the 19th
  International Conference on Artificial Intelligence and Statistics, {AISTATS}
  2016, Cadiz, Spain, May 9-11, 2016}, volume~51 of {\em {JMLR} Workshop and
  Conference Proceedings}, pages 473--481. JMLR.org, 2016.

\bibitem{LiWanLiu18}
C.~J. Li, M.~Wang, H.~Liu, and T.~Zhang.
\newblock Near-optimal stochastic approximation for online principal component
  estimation.
\newblock {\em Math. Program.}, 167(1, Ser. B):75--97, 2018.

\bibitem{LiTsi88}
C.-K. Li and N.-K. Tsing.
\newblock Some isometries of rectangular complex matrices.
\newblock {\em Linear and Multilinear Algebra}, 23(1):47--53, 1988.

\bibitem{MitCarJai13}
I.~Mitliagkas, C.~Caramanis, and P.~Jain.
\newblock Memory limited, streaming {PCA}.
\newblock In Burges et~al. \cite{DBLP:conf/nips/2013}, pages 2886--2894.

\bibitem{naor2012banach}
A.~Naor.
\newblock On the banach-space-valued azuma inequality and small-set
  isoperimetry of alon--roichman graphs.
\newblock {\em Combinatorics, Probability and Computing}, 21(4):623--634, 2012.

\bibitem{Oja82:Simplified-Neuron}
E.~Oja.
\newblock A simplified neuron model as a principal component analyzer.
\newblock {\em J. Math. Biol.}, 15(3):267--273, 1982.

\bibitem{Oja1985}
E.~Oja and J.~Karhunen.
\newblock On stochastic approximation of the eigenvectors and eigenvalues of
  the expectation of a random matrix.
\newblock {\em Journal of Mathematical Analysis and Applications},
  106(1):69--84, 1985.

\bibitem{Sa2015}
C.~D. Sa, C.~R{\'{e}}, and K.~Olukotun.
\newblock Global convergence of stochastic gradient descent for some non-convex
  matrix problems.
\newblock In F.~R. Bach and D.~M. Blei, editors, {\em Proceedings of the 32nd
  International Conference on Machine Learning, {ICML} 2015, Lille, France,
  6-11 July 2015}, volume~37 of {\em {JMLR} Workshop and Conference
  Proceedings}, pages 2332--2341. JMLR.org, 2015.

\bibitem{Sha15}
O.~Shamir.
\newblock Convergence of stochastic gradient descent for {PCA}.
\newblock In Balcan and Weinberger \cite{DBLP:conf/icml/2016}, pages 257--265.

\bibitem{Sha16}
O.~Shamir.
\newblock Fast stochastic algorithms for {SVD} and {PCA:} convergence
  properties and convexity.
\newblock In Balcan and Weinberger \cite{DBLP:conf/icml/2016}, pages 248--256.

\bibitem{Simchowitz2018}
M.~Simchowitz, A.~El~Alaoui, and B.~Recht.
\newblock Tight query complexity lower bounds for {PCA} via finite sample
  deformed {W}igner law.
\newblock In {\em S{TOC}'18---{P}roceedings of the 50th {A}nnual {ACM} {SIGACT}
  {S}ymposium on {T}heory of {C}omputing}, pages 1249--1259. ACM, New York,
  2018.

\bibitem{Tro12}
J.~A. Tropp.
\newblock User-friendly tail bounds for sums of random matrices.
\newblock {\em Found. Comput. Math.}, 12(4):389--434, 2012.

\bibitem{Tro15:Introduction-Matrix}
J.~A. Tropp.
\newblock An introduction to matrix concentration inequalities.
\newblock {\em Foundations and Trends in Machine Learning}, 8(1-2):1--230,
  2015.

\bibitem{Wedin1972}
P.-A. Wedin.
\newblock Perturbation bounds in connection with singular value decomposition.
\newblock {\em Nordisk Tidskrift for Informationsbehandling}, 12:99--111, 1972.

\end{thebibliography}

\end{document}